\documentclass[11pt]{article}

\usepackage{amsmath,amssymb,amsfonts,amsthm}
\usepackage[english]{babel}
\usepackage[utf8]{inputenc}
\usepackage[T1]{fontenc}
\usepackage{hyperref}
\usepackage{graphicx}
\usepackage{color}
\usepackage{comment}
\usepackage{wrapfig}
\usepackage{subcaption}
\usepackage{authblk}
\usepackage{thmtools}
\usepackage{paralist}
\usepackage{xcolor}
\usepackage{geometry}
\usepackage[linesnumbered, ruled]{algorithm2e}
\usepackage{paralist}
\usepackage{xspace}
\usepackage{tikz}
\usepackage{float}
\usepackage{thm-restate}

\definecolor{darkgreen}{rgb}{0,0.5,0}

\def\cupp{\stackrel{.}{\cup}}

\renewcommand{\ge}{\geqslant}
\renewcommand{\geq}{\geqslant}
\renewcommand{\le}{\leqslant}
\renewcommand{\leq}{\leqslant}

\def\Cscr{\mathcal{C}}

\def\Lscr{\mathcal{L}}

\def\Uscr{\mathcal{U}}

\def\cupp{\,\dot{\cup}\,}
\def\G2{G_{\!/\!/}}

\newtheorem*{theo*}{Theorem}
\newtheorem*{lem*}{Lemma}
\newtheorem*{prop*}{Proposition}
\newtheorem*{def*}{definition}
\newtheorem*{cor*}{Corollary}
\newtheorem*{conj*}{Conjecture}

\newtheorem{theorem}{Theorem}

\newtheorem{lemma}{Lemma}

\newtheorem{corollary}{Corollary}

\newtheorem{definition}{Definition}

\theoremstyle{definition}

\newtheorem{ex**}{Example}

\newcommand{\NNC}{\textsc{NonNegativeCycles}\xspace}
\newcommand{\NNCs}{\textsc{NNC}\xspace}
\newcommand{\EDPs}{\textsc{EDP}\xspace}
\newcommand\EDP{\textsc{EdgeDisjointPaths}\xspace}

\title{An Approximation Algorithm for \\ Fully Planar Edge-Disjoint Paths}

\author{Chien-Chung Huang\thanks{CNRS, École Normale Supérieure, Université PSL, Paris, \texttt{chien-chung.huang@ens.fr}}, Mathieu Mari\thanks{École Normale Supérieure, Université PSL, Paris,  \texttt{mathieu.mari@ens.fr}}, Claire Mathieu\thanks{CNRS, IRIF, Paris-Diderot, \texttt{claire.mathieu@irif.fr}}, Kevin Schewior\thanks{Universität zu Köln, \texttt{kschewior@gmail.com}. Supported by a DAAD PRIME grant.}, and Jens Vygen\thanks{Research Institute for Discrete Mathematics, Hausdorff Center for Mathematics, University of Bonn, \texttt{vygen@or.uni-bonn.de}}}

\begin{document}
%

\maketitle

\begin{abstract} 
We devise a constant-factor approximation algorithm for the maximization version of
the edge-disjoint paths problem if the supply graph together with the demand edges form a planar graph. 
By planar duality this is equivalent to packing cuts in a planar graph such that each cut contains exactly one demand edge. 
We also show that the natural linear programming relaxations have constant integrality gap, yielding an approximate max-multiflow min-multicut theorem.
\end{abstract}

\section{Introduction}

The edge-disjoint paths problem ($\EDPs$) is a fundamental problem in combinatorial optimization, 
consisting of connecting as many demand pairs as possible in a graph via edge-disjoint paths.
There is a large body of literature studying this problem in various settings. 
A primary goal has been to find conditions under which there is a solution satisfying all demands, 
e.g.\ when the \emph{cut condition} is sufficient; see Frank's survey \cite{Frank1990} or part VII of Schrijver's book \cite{schrijver}. 
Unfortunately many cases of $\EDPs$ are NP-hard, so it is natural to look for approximation algorithms. 
However there is no general theory and constant-factor approximations can only be expected in special cases.

One of the landmark results in this area is due to Seymour.
Let $G+H$ denote the union (of the edge sets) of the supply graph $G$ and the demand  graph $H$. 
Seymour~\cite{Seymour1981} proved that if $G+H$ is \emph{planar and  Eulerian}, 
the {cut condition} guarantees a solution to connect all demand pairs; such a solution can be found in polynomial time. 
Seymour's pioneering result has motivated a sequence of follow-up works investigating $\EDPs$ when $G+H$ is planar. We refer 
the readers to Frank's surveys~\cite{Frank1990,Frank1996} and Schrijver's book~\cite[in particular Chapter 74.2]{schrijver} for an overview of these results. 
For example, one can decide in polynomial time whether all
demand pairs can be connected when $G+H$ is planar and the demand pairs lie on a bounded number of faces of $G$~\cite{Midd1993}. 

Unfortunately the \emph{general} case that $G+H$ is planar is one of these cases in which EDP is NP-hard, as Middendorf and Pfeiffer~\cite{Midd1993} proved.
Very little seems to be known about approximation in that setting. 
Korach and Penn \cite{KorachPenn1992} showed that, given the cut condition, one can satisfy all demands except one for each face of $G$, and such a solution can be found in polynomial time. 
However, this does not imply an approximation ratio in general.

\subsection{Our Results}

Before we present our main results, let us define the edge-disjoint paths problem formally:

\begin{definition}\label{definition:EDP}
The \EDP problem \textnormal{(\EDPs)} takes as input a \emph{supply} graph $G=(V,E)$ and a \emph{demand} graph $H=(V,D)$ and asks for 
a maximum-cardinality set of pairwise edge-disjoint cycles such that each cycle consists of an edge $\{u,v\}$ in $D$ and a path between $u$ and $v$ in $E$. 
\end{definition}

The conditions that the set $\mathcal{C}$ of cycles in $G+H$ must satisfy can equivalently be written as: 
(1) $\forall C \in \mathcal{C}$, $|C \cap D|=1$, 
(2) 
$\forall e \in D\cup E$, 
$| \{C\mid C\in \mathcal{C}, C \ni e\}| \leq 1$.
We consider the case when $G+H$ is planar and present the first constant-factor approximation algorithm:

\begin{theorem}  
There is a polynomial-time $32$-approximation algorithm 
for \EDP if $G+H$ is planar. 
\label{thm:mainResultOne}
\end{theorem}

We prove this theorem in Sections \ref{section:roadmap}, \ref{section:rounding}, and \ref{section:jpacking}. 
Our proof is based on rounding a fractional solution to a certain LP relaxation to obtain a subset of demand edges for which the cut condition is satisfied, i.e., no cut has more of these demand edges than supply edges: in the planar dual, this translates into what we call the $\NNC$ problem. 
The $\NNC$ problem is NP-hard even when $G+H$ is planar (Section~\ref{section:NNCisNPcomplete}), but we give an approximation algorithm. 
Once this rounding is done and we have an approximate solution of $\NNC$, we can obtain edge-disjoint paths for at least half of these demand edges by building on ideas of Korach and Penn \cite{KorachPenn1992}.
We use the four color theorem in several places.

\medskip

Section \ref{section:LPs} analyzes the integrality gaps of various LPs. 
By comparing the output of our algorithm to the value
of the natural LP relaxation, we get an upper bound of 32 on its integrality gap. 
A special case, when every demand edge has an infinite (or large enough) number of parallel copies, has been called \emph{maximum integer multiflow}.
In this case the dual LP is a relaxation of the multicut problem,
which asks for a smallest set of supply edges whose deletion destroys all paths for all demand edges.
We show that this dual LP has an integrality gap of at most 2 if $G+H$ is planar. This yields:

\begin{restatable}{theorem}{mainresulttwo}
\label{thm:mainResultTwo}
If $G+H$ is planar, the minimum cardinality of a multicut is at most 64 times the maximum value of an integer multiflow.
\end{restatable}

 This is one of the few cases with a constant upper bound on the ratio of the smallest multicut and the largest \emph{integer} multiflow (see, e.g., \cite{Vazirani01}). 
Another such case is when $G$ arises from a tree by duplicating edges; then this ratio is 2 \cite{Garg1997}; 
in general the ratio can be as large as $\Theta(|D|)$, even when $G$ is planar.

\subsection{Further Related Work}

\paragraph{Hardness.} 
The decision version of $\EDPs$ is NP-complete~\cite{Karp75}, even in many special cases~\cite{Naves2008}.
In terms of approximation, $\EDPs$ is APX-hard~\cite{Chuzhoy05}.
Assuming that NP$\not\subseteq$DTIME$(n^{O(\log n)})$, where $n=|V|$,
there is no $2^{o(\sqrt{\log n})}$approximation for $\EDPs$, 
even when $G$ is planar and subcubic and each demand edge has one of its endpoints on the outer face of $G$~\cite{ChuzhoyKN17}. 
Assuming that for some $\delta>0$, not all problems in NP can be solved in randomized time $2^{n^\delta}$, 
there is no $n^{O(1/(\log\log n)^2)}$-approximation 
even when $G$ is a wall graph~\cite{ChuzhoyKN18}.

 The difficulty  is further illustrated by the integrality gap of a natural LP relaxation: 
 even when $G$ is planar and subcubic, the integrality gap is already in the order of $\Theta(\sqrt{n})$~\cite{Garg1997}. 

\paragraph{Positive results.} $\EDPs$ can be solved in polynomial time when the number of demand edges is bounded by a constant~\cite{ROBERTSON199565}. 
The best known approximation guarantee is $O(\sqrt{n})$~\cite{chekurikhannashepherd06},
even when $G$ is planar.

Nonetheless, there are some special cases for which 
significantly better approximation ratios are known: for instance, 
when $G$ is planar and Eulerian, or planar and 4-edge-connected, there is an $O(\log n)$-approximation~\cite{KawarabayashiK13}.  
When $G$ is a grid, there is an $O(\log^2 n)$-approximation~\cite{AumannR95}.  
When $G$ is a wall graph, there is a $\tilde{O}(n^{1/4})$-approximation~\cite{ChuzhoyK15}. 
When in addition all demand edges have one endpoint at the boundary of the wall, there is an $2^{O(\sqrt{\log n}\cdot \log\log n)}$-approximation~\cite{ChuzhoyKN18}. 
Yet none of those results gets to the range of a constant-factor approximation.

\paragraph{Variants.} One way to relax $\EDPs$ is to allow for congestion $c$, 
that is, in the solution, up to $c$ paths may share the same edge. 
It is known that EDP becomes significantly easier with congestion: with congestion 2, there is a polylogarithmic approximation~\cite{Chuzhoy12a}, 
and when in addition $G$ is planar, there is a constant-factor approximation~\cite{Seguin-CharbonneauS11}.

A closely-related but more difficult problem is the Node-Disjoint Paths (NDP) problem, in which 
two paths may not even share a common \emph{node}. 
$\EDPs$ can be reduced to NDP by taking the line graph of the supply graph, 
but this reduction does not preserve planarity.
If $G+H$ is planar, however, Middendorf and Pfeiffer showed how to reduce \EDPs to NDP while
preserving planarity. 
Hence the $\tilde{O}(n^{9/19})$-approximation algorithm for 
NDP with planar supply graphs~\cite{chuzhoy2016} implies the same approximation ratio for \EDPs with $G+H$ planar.

Another more general version of EDP is its directed version (one can reduce EDP to its directed version by replacing each edge by a simple gadget).
The directed version seems to be strictly harder than the undirected EDP; for example it is even NP-hard for two demand edges~\cite{FORTUNE1980111}.

\section{Preliminaries}

We recall here some notions and well-known properties related to cuts, cycles, planar duality, and $T$-joins. Many readers may want to skip this section.

\subsection{Cycles and Cuts}
Throughout this work we work with a (multi)graph $(V,D \cupp E)$ whose edge set is partitioned into two sets, $D$ and $E$. 
For example, we write $G+H=(V,D \cupp E)$, taking the union of the edge set of the supply graph $G=(V,E)$ and of the demand graph $H=(V,D)$. 
Each of these graphs can have parallel edges, and further parallel edges can arise by taking the (disjoint) union.
We assume throughout, without loss of generality, that $G+H$ is connected.

A \emph{cycle} (or circuit) in $(V,D\cupp E)$ is a set $C=\{e_1,\ldots,e_k\}$ of edges 
for which there is a set $V(C)=\{v_1,\ldots,v_k\}$ of $k$ distinct vertices such that $e_i$ has endpoints $v_i$ and $v_{i+1}$ 
for all $i=1,\ldots,k$ (where $v_{k+1}:=v_1$). 
If the graph has loops, every loop would also constitute a cycle.

A $D$-cycle is a cycle that contains exactly one edge of $D$.
Then $\EDP$ asks for a maximum number of pairwise disjoint \emph{$D$-cycles}.

We say an instance of $\EDPs$ has a \emph{complete solution} if there are $|D|$ pairwise disjoint $D$-cycles.
For the existence of a complete solution, a necessary condition is the \emph{cut condition}:
\begin{equation}
\label{eq:cutcondition}
|C\cap D|\le |C\cap E| \text{ for every cut }C.  
\end{equation}
It is necessary because every cycle intersects every cut in an even number of edges. 
In fact it is equivalent to requiring \eqref{eq:cutcondition} only 
for \emph{simple} cuts, i.e., edge sets $\delta(U)$ where both $U$ and $V\setminus U$ induce connected subgraphs. A cut is simple if and only if it is a minimal cut.

The cut condition is very useful although in general it is NP-hard to check (this is an unpublished observation of Seb{\H o}).
The reason is that in many special cases,
e.g.\ when $G+H$ is planar and Eulerian \cite{Seymour1981}, it is sufficient and then can also be checked in polynomial time.
However, the cut condition is not sufficient in general, even when $G+H$ is planar. 
For instance, if $G+H$ is $K_4$, the complete graph on four vertices, and $D$ is a perfect matching, then the cut condition is satisfied but there is no complete solution.

\subsection{Using Planarity}

We exploit planarity in two different ways.
First, we need an algorithmic version of the four color theorem.
Robertson, Sanders, Seymour, and Thomas \cite{Robertson1997} improved on the original
proof by Appel, Haken, and Koch and showed:

\begin{theorem}
Given a planar graph, there exists an $O(n^2)$-time algorithm to 
color its vertices with four colors such that no two adjacent vertices have the same color.
\label{thm:4color}
\end{theorem}

Second, we use planar duality. 
It is well known that for any connected planar graph $G$ there is a connected planar graph $G^*$ (a planar dual) 
such that there is a one-to-one-correspondence between the edges
of $G$ and the edges of $G^*$ that maps the cycles of $G$ to the simple cuts of $G^*$. 
Such a graph can be computed in polynomial time (see, e.g., \cite{MoharT01}). 
The planar dual is not always unique, but we can take an arbitrary one.
We will not distinguish between edges and edge sets and their images in the planar dual.

Let us call a simple cut that contains exactly one edge of $D$ a \emph{$D$-cut}. 
Then an equivalent formulation of $\EDPs$, if $G+H$ is planar, asks for finding a maximum number of pairwise disjoint $D$-cuts in the planar dual graph $(G+H)^*$.
The cut condition translates to the \emph{cycle condition} in the planar dual graph:
\begin{equation}
\label{eq:fullcyclecondition}
|C\cap D|\le |C\cap E| \text{ for every cycle }C.  
\end{equation}

Deleting an edge in a planar graph is equivalent to contracting the corresponding edge in the planar dual. 
These operations preserve planarity.

\subsection{\boldmath $T$-Joins}

Condition \eqref{eq:fullcyclecondition} is well known in the context of $T$-joins.
For an edge set $F$, let $\text{odd}(F)$ denote the set of vertices
whose degree is odd in $(V,F)$.
For a subset $T\subseteq V$, a \emph{$T$-join}
is a set $F$ of edges with $\text{odd}(F)=T$.
If $F$ is a minimum-cardinality $\text{odd}(F)$-join,
then $F$ is called a \emph{join}.
Guan \cite{Guan1962} observed that $F$ is a join if and only if 
\begin{equation}
\label{eq:guancyclecondition}
|C\cap F|\le |C\setminus F| \text{ for every cycle }C. 
\end{equation}
For any connected graph, any $T$ with $|T|$ even, and any real edge weights, a minimum-weight $T$-join can be computed in polynomial time \cite{EdmondsJohnson1973}.

If we have a set of pairwise disjoint $D$-cuts and $D'\subseteq D$ is the set of demand edges that belong to one of these cuts, then we must have
\begin{equation}
\label{eq:subsetcyclecondition}
|C\cap D'|\le |C\cap E| \text{ for every cycle }C.  
\end{equation}

This is equivalent to saying that $D'$ is a join after contracting the edges in $D\setminus D'$.
Although this condition is again not sufficient,
we will first find a set $D'\subseteq D$ satisfying \eqref{eq:subsetcyclecondition}
(in the planar dual graph).

\section{Roadmap}
\label{section:roadmap}

\subsection{Satisfying the Cycle Condition}

The preceding considerations motivate the following problem in the planar dual graph.

\begin{definition}\label{definition:NNC}
The \NNC problem (\NNCs) takes as input a supply graph $G=(V,E)$ and a demand graph $H=(V,D)$ 
and asks for a maximum-cardinality set $D'\subseteq D$ of demand edges such that
$$
|C\cap D'|\le |C\cap E| \text{ for every cycle } C \text{ in } G+H.
$$
\end{definition}

In other words, if edges in $E$ have weight $1$, edges in $D'$ have weight $-1$, and edges in $D\setminus D'$ have weight $0$, then every cycle must have nonnegative total weight. This problem is NP-hard even when $G+H$ is planar; 
see Section \ref{section:NNCisNPcomplete}. 
We devise a ${16}$-approximation algorithm if $G+H$ is planar.
It begins by solving the following linear programming relaxation,
which we call the \emph{nonnegative cycle LP}:
\begin{align}
\label{LP:NNC}
\max \sum_{e\in D} x_e & \\
\sum_{e \in C\cap D} x_e \le  |C\cap E|  &  \qquad \forall \text{ cycles } C \notag \\
0 \le x_e \le 1  & \qquad \forall e \in D \notag
\end{align}

\begin{lemma}
\label{lemma:solvelp}
The nonnegative cycle LP~(\ref{LP:NNC}) can be solved in polynomial time.
\end{lemma}

\begin{proof}
By the equivalence of optimization and separation \cite{GLS1981},
it is sufficient to solve the separation problem.
To this end, given $x\in\mathbb{R}^D$, 
define weights $w(e):=-x_e$ for $e\in D$ and $w(e):=1$ for $e\in E$.
Then the separation problem reduces to finding a negative-weight
cycle or deciding that there is none.
This can be done by computing a minimum-weight $\emptyset$-join $J$.
If $w(J)\ge 0$, then there is no negative-weight cycle.
Otherwise $J$ can be decomposed into cycles, and at least one of them will have negative weight. 
\end{proof}

The second step of our algorithm will round the LP solution.
In Section \ref{section:rounding} we will prove:

\begin{restatable}{theorem}{theoremstepone}
Let $x$ be a feasible solution to the nonnegative cycle LP~(\ref{LP:NNC}). 
Then there is a polynomial-time algorithm to construct an integral feasible solution 
that has value at least $\frac{1}{16}\sum_{e\in D} x_e$.
\label{thm:s1ands2}
\end{restatable}

\subsection{\boldmath Packing $T$-Cuts}

Let $x$ be an integral feasible solution to the nonnegative cycle LP
(in the planar dual), and let $J\subseteq D$ contain those demand edges $e$ for which $x_e=1$.
If we delete the other demand edges (or contract them in the planar dual), the resulting instance now satisfies
the cut condition (the cycle condition in the planar dual). 
Let $G'=(V,J\cupp E)^*$ be the graph that arises from the planar dual by contracting the edges in $D\setminus J$.
Then $J$ is a join in $G'$.
Our goal is now to find as many pairwise disjoint $J$-cuts in $G'$ as possible 
because after uncontracting, these will correspond to pairwise disjoint $D$-cuts in $(G+H)^*$ and hence to pairwise disjoint $D$-cycles in $G+H$.

Edmonds and Johnson \cite{EdmondsJohnson1973} and Lov\'asz \cite{Lovasz1975} showed that
a perfect half-integral packing of $J$-cuts exists
and can be computed in polynomial time:

\begin{theorem}
\label{thm:2packingofTcuts}
For every graph $G$ and every join $J$ in $G$
there exist vertex sets $U_1,\ldots,U_{2|J|}$ (not necessarily distinct)
that form a laminar family, such that $|\delta(U_i)\cap J|=1$ for all $i$ 
and for every edge $e$ of $G$ there are at most two indices $i$ with $e\in\delta(U_i)$.
Such sets can be computed in polynomial time.
\end{theorem}

This theorem has been called the $T$-cut packing theorem because if $J$ is a $T$-join (for some \emph{vertex} set $T$)
then the cuts will be $T$-cuts, i.e., cuts $\delta(U)$ with $|U\cap T|$ odd. However, we prefer to continue speaking of $J$-cuts because this is what we need. 
See \cite{Frank1996} and Theorem \ref{thm:algorithmicpackingofTcuts} for extensions of this result.
However, it is not easy (and without the planarity assumption in general impossible) to deduce from
this a large number of pairwise disjoint $J$-cuts.
Nevertheless we will show in Section \ref{section:jpacking}:

\begin{restatable}{theorem}{theoremquarter}
 Let $G$ be a planar graph and $J$ a join in $G$.
 Then there is a family of at least $\frac{|J|}{2}$ pairwise disjoint $J$-cuts in $G$, 
 and such a family can be computed in polynomial time.
 \label{thm:quarter}
\end{restatable}

\subsection{Overall Algorithm and Proof of Theorem \ref{thm:mainResultOne}}

We have now all ingredients to prove our main theorem.
The overall algorithm can be summarized as follows.
The input consists of two graphs $G=(V,E)$ and $H=(V,D)$ such that $G+H$ is planar.

\begin{enumerate}
\item[\emph{Step 0:}] Construct a planar dual graph $(G+H)^*$
\item[\emph{Step 1:}] Solve the nonnegative cycle LP \eqref{LP:NNC}
in the planar dual graph $(G+H)^*$, to get a fractional solution $x$; 
\item[\emph{Step 2:}] Use the algorithm of Theorem \ref{thm:s1ands2} (proved in Section~\ref{section:rounding}) 
to deduce an integral feasible solution $x'$ to \NNCs with $\sum_{e\in D}x'_e\ge \frac{1}{16}\sum_{e\in D}x_e$. 
Let $J:=\{e\in D:x'_e=1\}$.
\item[\emph{Step 3:}] Contract the edges in $D\setminus J$ (in the planar dual graph) 
and use the algorithm of Theorem \ref{thm:quarter} (proved in Section~\ref{section:jpacking}) to compute at least
$\frac{|J|}{2}$ pairwise disjoint $J$-cuts. 
By planar duality, these correspond to at least $\frac{|J|}{2}$ pairwise disjoint $D$-cycles in $G+H$.
Output these.
\end{enumerate}

We now prove our main result.

\begin{proof} (of Theorem~\ref{thm:mainResultOne}) \ 
A $J$-cut in $(V,J\cupp E)^*$ indeed corresponds to a $J$-cycle in $(V,J\cupp E)$ and hence in $G+H$. 
Therefore the output is a feasible solution. 
By Lemma \ref{lemma:solvelp}, Theorem~\ref{thm:s1ands2}, and Theorem~\ref{thm:quarter}, the algorithm runs in polynomial time. 
The number of $D$-cycles that it computes is 
at least $\frac{1}{32}$ times the LP value.

On the other hand, consider an optimum solution,
say consisting of $\text{OPT}$ many $D$-cycles.
Then there is a set $\bar D\subseteq D$ such that 
$|\bar D|=\text{OPT}$ and $(G,(V,\bar D))$ has a complete solution. This instance must therefore satisfy the cut condition, 
in other words we have
$|C\cap \bar D| \le |C\cap E|$ for every cut $C$ in $G+H$.
Therefore, in the planar dual, setting
$\bar x_e:=1$ for $e\in\bar D$ and $\bar x_e:=0$ for $e\in D\setminus\bar D$ defines a feasible solution to the nonnegative cycle LP. Hence the LP value is at least $\text{OPT}$.

We conclude that the algorithm is a polynomial-time ${32}$-approximation algorithm for $\EDPs$ if $G+H$ is planar. 
\end{proof}

\section{Rounding the Nonnegative Cycle LP (Proof of Theorem~\ref{thm:s1ands2})}\label{section:rounding}

In this section we prove Theorem \ref{thm:s1ands2}, whose statement we now recall. 

\theoremstepone*

\subsection{Algorithm} 

Starting from a feasible  solution $x$ of the nonnegative cycle LP (\ref{LP:NNC}), 
we first contract all cycles in $(V,D)$. Every edge $e$ that vanishes in this contraction has $x_e=0$, therefore this preprocessing will not change $\sum_{e\in D}x_e$.
Now $(V,D)$ is a forest.
Then we execute the following two algorithms and pick the better of the two solutions $S_1$ and $S_2$ obtained from Algorithms~\ref{algo:S1} and~\ref{algo:S2} respectively. 

The first algorithm is independent of the fractional solution $x$ and simply picks a subset of edges incident to leaves. It is the better choice if $(V,D)$ has many leaves (degree-1 vertices).

\medskip
\begin{algorithm}[H]\label{algo:S1}
\SetAlgoVlined

 \KwIn{an $\NNCs$ input $G+H$ such that $(V,D)$ is a forest, and a feasible solution $x$ to the nonnegative cycle LP $\eqref{LP:NNC}$.}
 \KwOut{
 a feasible solution $S_1$ to \NNC.}
Find a 4-coloring of $G$, consider a color class containing the largest number of leaves of $(V,D)$, and let $I$ be the set of leaves of that color.

Let $S_1$ be the set of edges in $D$ that are incident to a leaf in $I$.

    \Return $S_1$
\caption{The leaf algorithm}
\end{algorithm}

\medskip
The second algorithm does a delicate rounding of the fractional LP solution; see Figure \ref{fig:roundnonnegativecyclelpontree} for an example.

\medskip
\begin{algorithm}[H]\label{algo:S2}
\SetAlgoVlined
 \KwIn{an $\NNCs$ input $G+H$ such that $(V,D)$ is a forest, and a feasible solution $x$ to the nonnegative cycle LP $\eqref{LP:NNC}$.}
 \KwOut{a feasible solution $S_2$ to \NNC.}
For each connected component of the forest $(V,D)$:
Root the tree arbitrarily. 
Let $B$ denote the set of vertices with at least two children, and let $L$ denote the set of leaves.

For each $b'\in B\cup L$ and $b\in B\setminus\{b'\}$ such that $b$ is an ancestor of $b'$
and each inner vertex of the $b$-$b'$-path has exactly one child: 
subdivide the first edge on this path, thus inserting an artificial edge $a_{b,b'}$ incident to $b$, and set  $x_{a_{b,b'}}=0$. Let $A$ denote the set of artificial edges.

Define a budget function $y:D\cup A \to\mathbb{R}_{\ge 0}$.
Initially $y_e=x_e$ for all $e\in D\cup A$.

For each artifical edge $a_{b,b'}\in A$
follow the path from $b$ down to $b'$. 
For each traversed edge $e\in D$,
decrease $y_e$ and increase $y_{a_{b,b'}}$ by the same amount.
Do this until $y_{a_{b,b'}}=1$ or we reach $b'$. 

Consider the edges $e=(u,v)\in D$ in an order of non-increasing distance from the root, and for each such edge $e$, if $y_e>0$ then: 
Follow the path from $u$ up to the root. 
For each traversed edge $e'\in D\cup A$, decrease $y_{e'}$ and increase $y_e$ by the same amount.
Do this until $y_e=2$ or we reach the root. 

$S_2^0\gets \{ e\in D: y_e>0\} $.

Contract each tree of $(V,D)$, find a 4-coloring of the resulting graph, inducing a 4-coloring of the trees, 
choose the color class maximizing the number of edges of $S_2^0$ in trees of that color, and let $S_2$ denote that subset of $S_2^0$. 

    \Return $S_2$
\caption{The internal algorithm}
\end{algorithm}
\medskip

\begin{figure}
  \begin{center}
  \begin{tikzpicture}[scale=0.63]
  \tikzset{inner/.style={
  draw=black,thick,circle,inner sep=0em,minimum size=4pt,fill=black
  }}
  \tikzset{branch/.style={
  draw=red,thick,rectangle,inner sep=0em,minimum size=4pt,fill=red
  }}
  \tikzset{artificial/.style={
  very thick, densely dotted, red
  }}
  \tikzset{arc/.style={
  very thick
  }}
  \tikzset{take/.style={
  line width = .75mm, blue
  }}
    \tikzset{shade1/.style={
  line width=2mm, yellow
  }}
      \tikzset{shade2/.style={
  line width=2mm, green
  }}

\begin{scope}[xshift=-12cm]  
  \node at (6.8,20.1) {\scriptsize (a)};
  \node at (9.4,20.1) {\scriptsize root};
  \node[inner] (v0) at (10,20) {};
  \node[branch] (v1) at (10,19) {};
  \draw[arc] (v0) to node[right] {\scriptsize .8} (v1);

 \node[inner] (vk1) at (7,18) {};
  \draw[arc] (v1) to node[above] {\scriptsize .9 \ } (vk1);
  \node[inner] (vk2) at (7,17) {};
  \draw[arc] (vk1) to node[left] {\scriptsize .7} (vk2);
 \node[inner] (vk3) at (7,16) {};
  \draw[arc] (vk2) to node[left] {\scriptsize .9} (vk3);
  
  \node[inner] (vl1) at (9,18) {};
  \draw[arc] (v1) to node[left] {\scriptsize .3} (vl1);
  \node[inner] (vl2) at (9,17) {};
  \draw[arc] (vl1) to node[left] {\scriptsize .9} (vl2);
  \node[branch] (vl3) at (9,16) {};
  \draw[arc] (vl2) to node[left] {\scriptsize .4} (vl3);
 
  \node[inner] (vlr1) at (10,15) {};
  \draw[arc] (vl3) to node[right] {\scriptsize .8} (vlr1);
  
  \node[branch] (vll1) at (8,15) {};
  \draw[arc] (vl3) to node[left] {\scriptsize .3} (vll1);
  
  \node[inner] (vlll1) at (7,14) {};
  \draw[arc] (vll1) to node[left] {\scriptsize .5} (vlll1);
  
  \node[inner] (vllr1) at (9,14) {};
  \draw[arc] (vll1) to node[right] {\scriptsize .8} (vllr1);
  \node[inner] (vllr2) at (9,13) {};
  \draw[arc] (vllr1) to node[right] {\scriptsize .4} (vllr2);
     
  \node[inner] (vr1) at (11,18) {};
  \draw[arc] (v1) to node[right] {\scriptsize .6} (vr1);
  \node[inner] (vr2) at (11,17) {};
  \draw[arc] (vr1) to node[right] {\scriptsize .8} (vr2);
  \node[inner] (vr3) at (11,16) {};
  \draw[arc] (vr2) to node[right] {\scriptsize .9} (vr3);
  \node[inner] (vr4) at (11,15) {};
  \draw[arc] (vr3) to node[right] {\scriptsize .8} (vr4);
  \node[inner] (vr5) at (11,14) {};
  \draw[arc] (vr4) to node[right] {\scriptsize .1} (vr5);
\end{scope}

\begin{scope}[xshift=-6cm]  
  \node at (6.8,20.1) {\scriptsize (b)};
  \node at (9.4,20.1) {\scriptsize root};
  \node[inner] (v0) at (10,20) {};
  \node[branch] (v1) at (10,19) {};
  \draw[arc] (v0) to node[right] {\scriptsize .8} (v1);

  \node[inner] (vk0) at (7,19) {};
  \draw[artificial] (v1) to[bend right] node[below] {\scriptsize 0 \ \ } (vk0);
 \node[inner] (vk1) at (7,18) {};
  \draw[arc] (vk0) to node[left] {\scriptsize .9} (vk1);
  \node[inner] (vk2) at (7,17) {};
  \draw[arc] (vk1) to node[left] {\scriptsize .7} (vk2);
 \node[inner] (vk3) at (7,16) {};
  \draw[arc] (vk2) to node[left] {\scriptsize .9} (vk3);
  
  \node[inner] (vl0) at (9,19) {};
  \draw[artificial] (v1) to node[below] {\scriptsize 0} (vl0);
  \node[inner] (vl1) at (9,18) {};
  \draw[arc] (vl0) to node[left] {\scriptsize .3} (vl1);
  \node[inner] (vl2) at (9,17) {};
  \draw[arc] (vl1) to node[left] {\scriptsize .9} (vl2);
  \node[branch] (vl3) at (9,16) {};
  \draw[arc] (vl2) to node[left] {\scriptsize .4} (vl3);
 
   \node[inner] (vlr0) at (10,16) {};
  \draw[artificial] (vl3) to node[below] {\scriptsize 0} (vlr0);
  \node[inner] (vlr1) at (10,15) {};
  \draw[arc] (vlr0) to node[right] {\scriptsize .8} (vlr1);
  
  \node[inner] (vll0) at (8,16) {};
  \draw[artificial] (vl3) to node[below] {\scriptsize 0} (vll0);
  \node[branch] (vll1) at (8,15) {};
  \draw[arc] (vll0) to node[left] {\scriptsize .3} (vll1);
  
  \node[inner] (vlll0) at (7,15) {};
  \draw[artificial] (vll1) to node[below] {\scriptsize 0} (vlll0);
  \node[inner] (vlll1) at (7,14) {};
  \draw[arc] (vlll0) to node[left] {\scriptsize .5} (vlll1);
  
  \node[inner] (vllr0) at (9,15) {};
  \draw[artificial] (vll1) to node[below] {\scriptsize 0} (vllr0);
  \node[inner] (vllr1) at (9,14) {};
  \draw[arc] (vllr0) to node[right] {\scriptsize .8} (vllr1);
  \node[inner] (vllr2) at (9,13) {};
  \draw[arc] (vllr1) to node[right] {\scriptsize .4} (vllr2);
     
  \node[inner] (vr0) at (11,19) {};
  \draw[artificial] (v1) to node[below] {\scriptsize 0} (vr0); 
  \node[inner] (vr1) at (11,18) {};
  \draw[arc] (vr0) to node[right] {\scriptsize .6} (vr1);
  \node[inner] (vr2) at (11,17) {};
  \draw[arc] (vr1) to node[right] {\scriptsize .8} (vr2);
  \node[inner] (vr3) at (11,16) {};
  \draw[arc] (vr2) to node[right] {\scriptsize .9} (vr3);
  \node[inner] (vr4) at (11,15) {};
  \draw[arc] (vr3) to node[right] {\scriptsize .8} (vr4);
  \node[inner] (vr5) at (11,14) {};
  \draw[arc] (vr4) to node[right] {\scriptsize .1} (vr5);
\end{scope}

\begin{scope}[xshift=0cm]  
  \node at (6.8,20.1) {\scriptsize (c)};
  \node at (9.4,20.1) {\scriptsize root};
  \node[inner] (v0) at (10,20) {};
  \node[branch] (v1) at (10,19) {};
  \draw[arc] (v0) to node[right] {\scriptsize .8} (v1);
  
  \node[inner] (vk0) at (7,19) {};
  \draw[artificial] (v1) to[bend right] node[below] {\scriptsize 1 \ \ } (vk0);
 \node[inner] (vk1) at (7,18) {};
  \draw[arc] (vk0) to node[left] {\scriptsize 0} (vk1);
  \node[inner] (vk2) at (7,17) {};
  \draw[arc] (vk1) to node[left] {\scriptsize .6} (vk2);
 \node[inner] (vk3) at (7,16) {};
  \draw[arc] (vk2) to node[left] {\scriptsize .9} (vk3);
    
  \node[inner] (vl0) at (9,19) {};
  \draw[artificial] (v1) to node[below] {\scriptsize 1} (vl0);
  \node[inner] (vl1) at (9,18) {};
  \draw[arc] (vl0) to node[left] {\scriptsize 0} (vl1);
  \node[inner] (vl2) at (9,17) {};
  \draw[arc] (vl1) to node[left] {\scriptsize .2} (vl2);
  \node[branch] (vl3) at (9,16) {};
  \draw[arc] (vl2) to node[left] {\scriptsize .4} (vl3);
 
   \node[inner] (vlr0) at (10,16) {};
  \draw[artificial] (vl3) to node[below] {\scriptsize .8} (vlr0);
  \node[inner] (vlr1) at (10,15) {};
  \draw[arc] (vlr0) to node[right] {\scriptsize 0} (vlr1);
  
  \node[inner] (vll0) at (8,16) {};
  \draw[artificial] (vl3) to node[below] {\scriptsize .3} (vll0);
  \node[branch] (vll1) at (8,15) {};
  \draw[arc] (vll0) to node[left] {\scriptsize 0} (vll1);
  
  \node[inner] (vlll0) at (7,15) {};
  \draw[artificial] (vll1) to node[below] {\scriptsize .5} (vlll0);
  \node[inner] (vlll1) at (7,14) {};
  \draw[arc] (vlll0) to node[left] {\scriptsize 0} (vlll1);
  
  \node[inner] (vllr0) at (9,15) {};
  \draw[artificial] (vll1) to node[below] {\scriptsize 1} (vllr0);
  \node[inner] (vllr1) at (9,14) {};
  \draw[arc] (vllr0) to node[right] {\scriptsize 0} (vllr1);
  \node[inner] (vllr2) at (9,13) {};
  \draw[arc] (vllr1) to node[right] {\scriptsize .2} (vllr2);
     
  \node[inner] (vr0) at (11,19) {};
  \draw[artificial] (v1) to node[below] {\scriptsize 1} (vr0); 
  \node[inner] (vr1) at (11,18) {};
  \draw[arc] (vr0) to node[right] {\,\scriptsize 0} (vr1);
  \node[inner] (vr2) at (11,17) {};
  \draw[arc] (vr1) to node[right] {\scriptsize .4} (vr2);
  \node[inner] (vr3) at (11,16) {};
  \draw[arc] (vr2) to node[right] {\scriptsize .9} (vr3);
  \node[inner] (vr4) at (11,15) {};
  \draw[arc] (vr3) to node[right] {\scriptsize .8} (vr4);
  \node[inner] (vr5) at (11,14) {};
  \draw[arc] (vr4) to node[right] {\scriptsize .1} (vr5);
\end{scope}

\begin{scope}[xshift=6cm]  

  \draw[shade1] (8.5,19.48) to [out=180,in=30] (7,19);
  \draw[shade1] (7,19) to (7,16);
  \draw[shade2] (11,14) to (11,17.5);
  \draw[shade1] (11,17.5) to (11,19);
  \draw[shade1] (11,19) to (10.1,19);
  \draw[shade1] (9,13) to (9,15);
  \draw[shade1] (9,15) to (8,15);
  \draw[shade1] (8,15) to (8,16);
  \draw[shade1] (8,16) to (9,16);
  \draw[shade1] (9,16) to (9,17.5);
  \draw[shade2] (9,17.5) to (9,19);
  \draw[shade2] (9,19) to (10,19);
  \draw[shade2] (10,19) to (10,20);

  \node at (6.8,20.1) {\scriptsize (d)};
  \node at (9.4,20.1) {\scriptsize root};
  \node[inner] (v0) at (10,20) {};
  \node[branch] (v1) at (10,19) {};
  \draw[arc] (v0) to node[right] {\scriptsize 0} (v1);
  
  \node[inner] (vk0) at (7,19) {};
  \draw[artificial] (v1) to[bend right] node[below] {\scriptsize .5 \ \ } (vk0);
 \node[inner] (vk1) at (7,18) {};
  \draw[arc] (vk0) to node[left] {\scriptsize 0} (vk1);
  \node[inner] (vk2) at (7,17) {};
  \draw[arc] (vk1) to node[left] {\scriptsize 0} (vk2);
 \node[inner] (vk3) at (7,16) {};
  \draw[take] (vk2) to node[right] {\scriptsize $e_3$} node[left] {\scriptsize 2} (vk3);  
  
  \node[inner] (vl0) at (9,19) {};
  \draw[artificial] (v1) to node[below] {\scriptsize 0} (vl0);
  \node[inner] (vl1) at (9,18) {};
  \draw[arc] (vl0) to node[left] {\scriptsize 0} (vl1);
  \node[inner] (vl2) at (9,17) {};
  \draw[take] (vl1) to node[right] {\scriptsize $e_4$} node[left] {\scriptsize 1.9} (vl2);
  \node[branch] (vl3) at (9,16) {};
  \draw[arc] (vl2) to node[left] {\scriptsize 0} (vl3);
 
   \node[inner] (vlr0) at (10,16) {};
  \draw[artificial] (vl3) to node[below] {\scriptsize .8} (vlr0);
  \node[inner] (vlr1) at (10,15) {};
  \draw[arc] (vlr0) to node[right] {\scriptsize 0} (vlr1);
  
  \node[inner] (vll0) at (8,16) {};
  \draw[artificial] (vl3) to node[below] {\scriptsize 0} (vll0);
  \node[branch] (vll1) at (8,15) {};
  \draw[arc] (vll0) to node[left] {\scriptsize 0} (vll1);
  
  \node[inner] (vlll0) at (7,15) {};
  \draw[artificial] (vll1) to node[below] {\scriptsize .5} (vlll0);
  \node[inner] (vlll1) at (7,14) {};
  \draw[arc] (vlll0) to node[left] {\scriptsize 0} (vlll1);
  
  \node[inner] (vllr0) at (9,15) {};
  \draw[artificial] (vll1) to node[below] {\scriptsize 0} (vllr0);
  \node[inner] (vllr1) at (9,14) {};
  \draw[arc] (vllr0) to node[right] {\scriptsize 0} (vllr1);
  \node[inner] (vllr2) at (9,13) {};
  \draw[take] (vllr1) to node[left] {\scriptsize $e_1$} node[right] {\scriptsize 2} (vllr2);
     
  \node[inner] (vr0) at (11,19) {};
  \draw[artificial] (v1) to node[below] {\scriptsize 0} (vr0); 
  \node[inner] (vr1) at (11,18) {};
  \draw[arc] (vr0) to node[right] {\,\scriptsize 0} (vr1);
  \node[inner] (vr2) at (11,17) {};
  \draw[take] (vr1) to node[left] {\scriptsize $e_5$} node[right] {\scriptsize 1.2} (vr2);
  \node[inner] (vr3) at (11,16) {};
  \draw[arc] (vr2) to node[right] {\scriptsize 0} (vr3);
  \node[inner] (vr4) at (11,15) {};
  \draw[arc] (vr3) to node[right] {\scriptsize 0} (vr4);
  \node[inner] (vr5) at (11,14) {};
  \draw[take] (vr4) to node[left] {\scriptsize $e_2$} node[right] {\scriptsize 2} (vr5);
\end{scope}

  \end{tikzpicture}
  \end{center}    
  \caption{Illustrating Algorithm \ref{algo:S2}. 
  (a): the original tree with an arbitrarily chosen root; vertices of $B$ are red squares; next to each edge $e$ the value $x_e$ is shown. 
  (b): the rooted tree after inserting the artificial edges (red, dotted) and the initial budgets $y_e$ (after step 3). 
  (c): each artificial edge has collected up to one unit from below (in step 4).
  (d): edges $e_1,\ldots,e_5$ (shown in blue, bold) are collecting up to two units from above (in step 5) and are included into $S_2^0$ (in step 6). 
  We see the final distribution if the edges are considered in this order. 
  The green and yellow shades show from where the final budgets $y_{e_1},\ldots,y_{e_5}$ come from.
  \label{fig:roundnonnegativecyclelpontree}
  }
\end{figure}
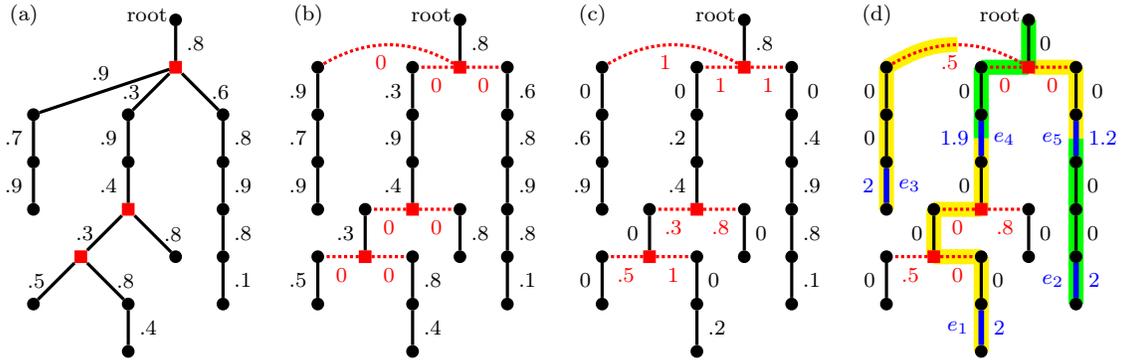

\subsection{Analysis}

These algorithms are well-defined and polynomial time by Theorem \ref{thm:4color}. 
We will prove that $S_1$ and $S_2$ are feasible $\NNC$ solutions (Lemmas~\ref{lemma:S1feasible} and~\ref{lemma:S2feasible}, respectively) 
and that  $\max \{ |S_1|, |S_2|\} \geq \frac{1}{16} \sum_{e\in D} x_e$ (Lemma~\ref{lemma:costtheorem1}), thus establishing Theorem \ref{thm:s1ands2}.

\begin{lemma}
$S_1$ is a feasible solution for \NNC. 
\label{lemma:S1feasible}
\end{lemma}

\begin{proof}
Let $C$ be a cycle in $G+H$. By definition of $S_1$ each edge of $C\cap S_1$ is incident to at least one vertex of $I$, 
and since vertices of $I$ are leaves of $(V,D)$ those vertices are all distinct, so  $|C\cap S_1|\leq |V(C)\cap I|$. 
Each vertex of $V(C)\cap I$ is incident to at most one edge of $C\cap D$, hence to at least one edge of $C\cap E$. 
Since $I$ is an independent set in $G$, those edges are all distinct, so $|V(C)\cap I|\leq |C\cap E|$.
\end{proof}

To prove that $S_2$ is also a feasible solution for \NNC, we first show the following key inequality:

\begin{lemma}
\label{lemma:keyinequality}
Let $P$ be the set of edges of a path in $(V,D)$. Then
\begin{equation}
\label{eq:keyinequality}
    \sum_{e\in P} x_e\ge 2 |S_2\cap P|-2.
\end{equation}
\end{lemma}
\begin{proof}
We may assume that $S_2\cap P\not=\emptyset$, for the assertion is trivial otherwise.
Let $v$ be the vertex in $P$ that is closest to the root (picked in step 1 of Algorithm \ref{algo:S2}) of the connected component of $(V,D)$ that contains $P$.
Then $P$ is the union of two paths $P_1$ and $P_2$ that both begin in $v$ and go down the tree. One of these paths may be empty.
Let $i\in\{1,2\}$ with $S_2\cap P_i\not=\emptyset$, and let $e_i^h$ be the edge of $S_2\cap P_i$ that is closest to $v$ (and to the root).
After step 5 of Algorithm \ref{algo:S2}, every edge $e$ of $(S_2\cap P_i)\setminus\{e_i^h\}$ has $y_e=2$, and these budgets come from
edges of $P_i$. The latter is because
\begin{itemize}
    \item[(i)] $e_i^h$ is the only edge in $S_2\cap P_i$ to which budget that is originally from closer to the root than $v$ can be moved, and
    \item[(ii)] budget that was moved towards the root (to an artificial edge) in step 4 remains at this artificial edge or is moved in step 5 to an edge that is not closer to the root than the edge where the budget was initially.
\end{itemize} Thus,
$$\sum_{e\in P_i}x_e \ge \sum_{e\in S_2\cap P_i, e\neq e_i^h} y_e= 2(|S_2\cap P_i|-1).$$
This already implies \eqref{eq:keyinequality} if $S_2\cap P=S_2\cap P_i$.

It remains to consider the case that neither $S_2\cap P_1$ nor $S_2\cap P_2$ is empty. 
Then $v\in B$, and after step 2 there is an artificial edge incident to $v$ on the path from $v$ to $e_i^h$  (for each $i\in\{1,2\}$).
Let $f_i$ be the last artificial edge on the path from $v$ to $e_i^h$. 
Since $e_i^h\in S_2$, we had $y_{e_i^h}>0$ after step 4 (i.e., even before we started augmenting $y_{e_i^h}$ in step 5). 
Therefore $y_{f_i}=1$ after step 4 (and, by the same argument as before, this budget comes from edges of $P_i$), and $y_{e_i^h}\ge 1$ after step 5. 
Thus
$$\sum_{e\in P_i} x_e\ge 1 + \sum_{e\in S_2\cap P_i, e\neq e_i^h} y_e = 1 + 2 (|S_2\cap P_i|-1) = 2 |S_2\cap P_i| - 1.$$
Adding this for $P_1$ and $P_2$ yields \eqref{eq:keyinequality}.
\end{proof}

\begin{lemma}
$S_2$ is a feasible solution for \NNC.
\label{lemma:S2feasible}
\end{lemma}
\begin{proof}
Let $C$ be a cycle in $G+H$ that contains at least one demand edge and at least one supply edge 
(recall that $(V,D)$ is a forest).
The cycle $C$ alternates between $D$-segments (maximal subpaths of $C$ all whose edges belong to $D$) and $E$-segments. 
Suppose there is some cycle $C$ with $|C\cap S_2|>|C\cap E|$. 
Among those cycles $C$, choose one such that $C$ has as few $D$-segments as possible.

We claim that for every tree $T$ of $D$, $C\cap T$ is connected. 
Assume not; then there is a path $P\subseteq D$ whose endpoints $u$ and $v$ belong to $V(C)$,
but no edge of $P$ belongs to $C$.
The two endpoints $u$ and $v$ of $P$ partition $C$ into two $u$-to-$v$ paths, $P_1$ and $P_2$. 
One of the two cycles $C_1=P_1\cup P$ and $C_2=P_2\cup P$ must have $|C_i\cap S_2|>|C_i\cap E|$ 
and has strictly fewer $D$-segments than $C$, contradicting the choice of $C$. This proves the claim. 

Thus each $D$-segment of $C$ belongs to a different tree.
By Lemma \ref{lemma:keyinequality},
every $D$-segment $P$ satisfies \eqref{eq:keyinequality}.
Moreover, since $x$ is a feasible solution to the nonnegative cycle LP, $\sum_{e\in C\cap D}x_e \le |C\cap E|$.

If $|C\cap E|=1$, then \eqref{eq:keyinequality} yields 
$2|C\cap S_2| \le 2+\sum_{e\in C\cap D} x_e \le 2+ |C\cap E| = 3$, implying $|C\cap S_2|\le 1=|C\cap E|$. 
Otherwise, thanks to selecting an independent set in the contracted graph in step 7, if there is more than one $D$-segment, every $E$-segment has at least two edges.
Summing \eqref{eq:keyinequality} over all $D$-segments $P$ yields
$2|C\cap S_2| \le |C\cap E|+ \sum_{e\in C\cap D}x_e \le |C\cap E| + |C\cap E|$. Therefore, 
$ |C\cap S_2| \le |C\cap E|$ as required. 
\end{proof}

\begin{lemma}\label{lemma:costtheorem1}
$$\max \{ |S_1|, |S_2|\} \geq \frac{1}{16} \sum_{e\in D} x_e.$$ 
\end{lemma}

\begin{proof}
The construction of Algorithm~\ref{algo:S2} maintains the
invariant $\sum_{e\in D} x_e=\sum_{e\in D\cup A} y_e$. 
Since edges of $S_2^0$ have $y_e\leq 2$, artificial edges have $y_e\leq 1$, and other edges have $y_e=0$, we can write $\sum_{e\in D\cup A} y_e\leq 2|S_2^0|+|A|$. 
Observe that the number of artificial edges is at most $2|L|-2$. Finally, by construction $|L|\leq 4|S_1|$ and $|S_2^0|\leq 4|S_2|$. 

Combining, we conclude
$$\sum_{e\in D} x_e\leq 2|S_2^0|+2|L|\leq 8 (|S_2|+ |S_1|) \leq 16 \max\{|S_1|,|S_2|\}.$$
\end{proof}

%
\section{Packing Cuts (Proof of Theorem~\ref{thm:quarter})}\label{section:jpacking}
%
In this section we prove Theorem~\ref{thm:quarter}, whose statement we now recall. 

\theoremquarter*

Recall that a \emph{$J$-cut} is a simple cut that contains exactly one edge from $J$.
An example by Korach and Penn \cite{KorachPenn1992} shows that the factor 2 in Theorem \ref{thm:quarter} is best possible.
Without planarity we cannot hope for any constant factor, e.g.\ if $G$ is a complete graph and $J$ a perfect matching in $G$, then
there are no two disjoint $J$-cuts.

Let $J$ be a join in a planar graph $G$.
Middendorf and Pfeiffer \cite{Midd1993} showed that it is NP-complete to decide whether there are $|J|$ pairwise disjoint $J$-cuts. 
Hence also 
finding the maximum number of pairwise disjoint $J$-cuts is NP-hard.
However, we give an approximation algorithm in the following.
Our proof is inspired by the paper of Korach and Penn  \cite{KorachPenn1992}.

Recall that a family of subsets of $V$ is \emph{laminar} if for any two of those sets either they are disjoint or one is a subset of the other.
It is \emph{cross-free} if for any two of those sets, either they are disjoint, or one is a subset of the other, or their union is $V$.
We will show the following:

\begin{lemma}
\label{lemma:cutpackingfromhalfintegral}
Let $G=(V,E)$ be a connected planar graph.
Let $\Lscr$ be a laminar family of (not necessarily distinct) vertex sets of $G$ such that 
$\delta(U)$ is a simple cut for all $U\in\Lscr$ and every edge is contained in at most two of these cuts. 
Then there exists a polynomial-time algorithm to compute a sub-family $\Lscr'\subseteq\Lscr$ such that $|\Lscr'|\ge\frac{1}{4}|\Lscr|$ 
and the cuts $\delta(U)$ for $U\in\Lscr'$ are pairwise disjoint.
\end{lemma}

 Again, the example by Korach and Penn \cite{KorachPenn1992} shows that the factor 4 in Lemma \ref{lemma:cutpackingfromhalfintegral} is best possible.
Without planarity we cannot hope for any constant factor, e.g.\ if $G$ is a complete graph and 
$\Lscr$ contains all singletons.

Before we prove Lemma \ref{lemma:cutpackingfromhalfintegral}, let us first see how it implies Theorem \ref{thm:quarter}.

\begin{proof}(Theorem \ref{thm:quarter}) \ 
Let $J$ be a join in $G$.
By Theorem \ref{thm:2packingofTcuts}, we can compute vertex sets $U_1,\ldots,U_{2|J|}$ (not necessarily distinct)
that form a laminar family $\Lscr$, such that $|\delta(U_i)\cap J|=1$ for all $i$
and every edge is contained in at most two of the cuts $\delta(U_i)$, $i=1,\ldots,2|J|$. 
Before applying Lemma \ref{lemma:cutpackingfromhalfintegral},  
we want to make sure that these
cuts are all simple cuts, and this will be achieved in the following by a sequence of transformations. 

For $i=1,\ldots,2|J|$ let $U_i'\subseteq U_i$ be the vertex set of the connected component of $G[U_i]$ 
for which $|\delta(U_i')\cap J|=1$.
Then $\delta(U_i')\subseteq \delta(U_i)$.
We claim that the sets $U_1',\ldots,U_{2|J|}'$ form a laminar family $\Lscr'$.
Let $1\le i<j\le 2|J|$.
If $U_i\cap U_j=\emptyset$, then $U_i'\cap U_j'=\emptyset$.
If $U_i\subseteq U_j$, then the vertex sets of the connected components of $G[U_i]$ are subsets 
of the vertex sets of the connected components of $G[U_j]$, so $U_i'$ is either disjoint from $U_j'$ or a subset of $U_j'$.
The case $U_j\subseteq U_i$ is symmetric.

Now $G[U_i']$ is connected for all $i$. 
For $i=1,\ldots,2|J|$ let $U_i''\subseteq V\setminus U_i'$ be the vertex set of the connected component of $G[V\setminus U_i']$ 
for which $|\delta(U_i'')\cap J|=1$.
Then $\delta(U_i'')\subseteq \delta(U_i')\subseteq\delta(U_i)$.
Note that for each $i$, $G[U_i'']$ and $G[V\setminus U_i'']$ are connected and $|\delta(U_i'')\cap J|=1$.
In other words, $\delta(U_1''),\ldots,\delta(U_{2|J|}'')$ are $J$-cuts.

We claim that the sets $U_1'',\ldots,U_{2|J|}''$ form a cross-free family $\Lscr''$.
Let $1\le i<j\le 2|J|$.
If $U_i'\subseteq U_j'$, then the vertex sets of the connected components of $G[V\setminus U_j']$ are subsets 
of the vertex sets of the connected components of $G[V\setminus U_i']$, so $U_j''$ is either disjoint from $U_i''$ or a subset of $U_i''$.
The case $U_j'\subseteq U_i'$ is symmetric.

It remains to consider the case $U_i'\cap U_j'=\emptyset$.
Suppose $X:=V\setminus (U_i''\cup U_j'')$ is nonempty.
Since $G$ is connected, we have 
$\emptyset\not= \delta(X) = \delta(U_i''\cup U_j'') \subseteq \delta(U_i'\setminus U_j'') \cup \delta(U_j'\setminus U_i'')$, where the last step is because all edges with exactly one endpoint in $U_i''$ have the other endpoint in $U_i'$, and all edges with exactly one endpoint in $U_j''$ have the other endpoint in $U_j'$.
Hence at least one of $U_i'\setminus U_j''$ or $U_j'\setminus U_i''$ is nonempty.
Suppose, without loss of generality, $U_i'\setminus U_j''\not=\emptyset$.
Because $G[U_i']$ is connected, $U_i'$ is a subset of the vertex set of a connected component of $G[V\setminus U_j']$.
So $U_i'\cap U_j''=\emptyset$.
But then $U_j''$ is a subset of the vertex set of a connected component of $G[V\setminus U_i']$, and thus it is a subset of $U_i''$ or disjoint from $U_i''$. 

So indeed $\Lscr''$ is cross-free.
To obtain a laminar family, choose $v \in V$ arbitrarily, and for every $U''\in\Lscr''$ with $v\in U''$, replace $U''$ by $V\setminus U''$.
We obtain a laminar family $\Lscr'''$ of $2|J|$ (not necessarily distinct) vertex sets such that $\delta(U)$ is a $J$-cut for all $U\in\Lscr'''$
and every edge is contained in at most two of these cuts.
It is obvious that all the above steps can be performed in polynomial time.
Applying Lemma \ref{lemma:cutpackingfromhalfintegral} concludes the proof.
 \end{proof}

To prove Lemma \ref{lemma:cutpackingfromhalfintegral}, consider the following recursive algorithm:

\medskip
\begin{algorithm}[H]\label{algo:packing}
\SetAlgoVlined
 \KwIn{a planar graph $G=(V,E)$ and a laminar family $\Lscr$ of (not necessarily distinct) nonempty subsets of $V$ such that 
$G[U]$ and $G[V\setminus U]$ are connected for all $U\in\Lscr$ and 
every edge is contained in at most two of the cuts $\delta(U)$, $U\in\Lscr$. }
 \KwOut{a subset $\Lscr'\subseteq \Lscr$ such that the cuts $\delta(U)$, $U\in\Lscr'$, are pairwise disjoint}
 \If(\tcc*[h]{case 1}){the elements of $\Lscr$ are pairwise disjoint}{
 Consider the planar graph $P$ obtained by contracting each $U\in \Lscr$ into a single vertex and deleting all other vertices. Compute a 4-coloring on $P$.
 
 \Return a subset of $\Lscr$ that corresponds to a maximum-cardinality color class
 }
 \If(\tcc*[h]{case 2}){there exist $U_1,U_2\in\Lscr$ such that $U_1=U_2$}{
    \Return $\{U_1\}\cup \textsc{CutPacking}(G,\Lscr\setminus\{U_1,U_2\})$
 }
\tcc{case 3}
 Let $\overline{U}\in \Lscr$ be a set that is not minimal but all sets $U_1,\dots,U_l\in \Lscr$ that are proper subsets of $\overline{U}$ are minimal (inclusionwise).
 
Consider the planar graph $P$ obtained by contracting each $U_i$ into a \emph{normal} vertex, 
deleting all other vertices in $\bar U$, and contracting all vertices outside $\overline{U}$ into a single \emph{special} vertex.

Find a 4-coloring of $P$ and choose a color class $\mathcal{K}$ with the largest number of normal vertices. 
If all color classes contain $\frac{l}{4}$ normal vertices, choose the one that contains the special vertex.
Let $\Lscr'\subset \Lscr$ be the set of $U\in\Lscr$ that correspond to normal vertices in $\mathcal{K}$.

\uIf(\tcc*[h]{case 3a}){$\mathcal{K}$ contains more than $\frac{l}{4}$ normal vertices}{
\Return $\Lscr'\cup \textsc{CutPacking}(\Lscr\setminus \{\overline{U}, U_1,\dots,U_l\})$ }
\Else(\tcc*[h]{case 3b}){
    \Return $\Lscr'\cup \textsc{CutPacking}(\Lscr\setminus \{U_1,\dots,U_l\})$}
\caption{{\sc CutPacking} \label{alg:cutpacking}}
\end{algorithm}

\begin{proof}(Lemma \ref{lemma:cutpackingfromhalfintegral}) \ 
We show the assertion by induction on $|\Lscr|$. 
We say here that a subset $\Lscr'\subseteq \Lscr$ is \emph{stable} if all $\delta(U)$ for $U\in \Lscr'$ are pairwise disjoint. 
We apply Algorithm \ref{alg:cutpacking}.
Note that the graphs $P$ in step 2 and step 7 are indeed planar because we only contract connected vertex sets.
By Theorem~\ref{thm:4color}, the algorithm runs in polynomial time.
We start with the base case of our induction.

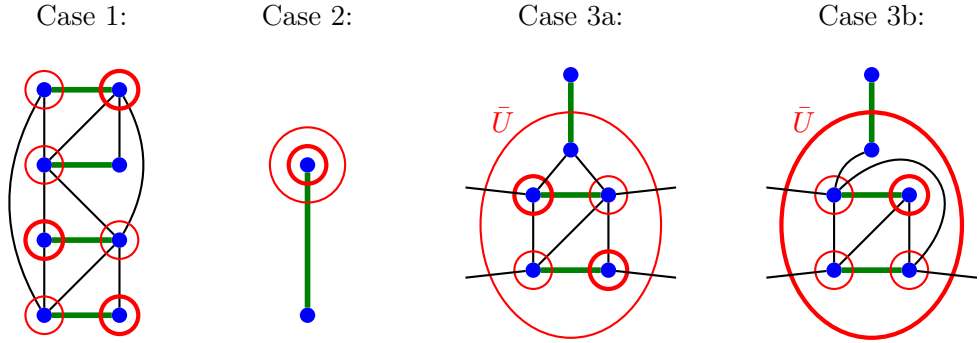
\begin{figure}
  \begin{center}
  \begin{tikzpicture}[scale=1]
  \tikzset{vertex/.style={
  draw=blue,thick,circle,inner sep=0em,minimum size=5pt,fill=blue
  }}
  \tikzset{edge/.style={
  thick
  }}
  \tikzset{j/.style={
  line width = .75mm, darkgreen
  }}
    \tikzset{cut/.style={
  thick, red
  }}
   \tikzset{take/.style={
  ultra thick, red
  }}

 \begin{scope}[xshift=0cm]  
  \node at (0.5,4) {Case 1:};
  \node[vertex] (v00) at (0,0) {};
  \node[vertex] (v10) at (1,0) {};
  \draw[j] (v00) to (v10);
  \node[vertex] (v01) at (0,1) {};
  \node[vertex] (v11) at (1,1) {};
  \draw[j] (v01) to (v11);
  \node[vertex] (v02) at (0,2) {};
  \node[vertex] (v12) at (1,2) {};
  \draw[j] (v02) to (v12);
  \node[vertex] (v03) at (0,3) {};
  \node[vertex] (v13) at (1,3) {};
  \draw[j] (v03) to (v13);
  \draw[edge] (v00) to (v01);
  \draw[edge] (v01) to (v02);
  \draw[edge] (v02) to (v03);
  \draw[edge] (v10) to (v11);
  \draw[edge] (v12) to (v13);
  \draw[edge] (v00) to (v11);
  \draw[edge] (v11) to (v02);
  \draw[edge] (v02) to (v13);
  \draw[edge] (v11) to[bend right] (v13);
  \draw[edge] (v03) to[bend right] (v00);
  \draw[cut] (0,0) circle (.25 and .25);
  \draw[take] (0,1) circle (.25 and .25);
  \draw[cut] (0,2) circle (.25 and .25);
  \draw[cut] (0,3) circle (.25 and .25);
  \draw[take] (1,0) circle (.25 and .25);
  \draw[cut] (1,1) circle (.25 and .25);
  \draw[take] (1,3) circle (.25 and .25);
 \end{scope} 
  
 \begin{scope}[xshift=3cm]  
  \node at (0.5,4) {Case 2:};
  \node[vertex] (w0) at (0.5,0) {};
  \node[vertex] (w1) at (0.5,2) {};
  \draw[j] (w0) to (w1); 
  \draw[cut] (0.5,2) circle (.5 and .5);
  \draw[take] (0.5,2) circle (.25 and .25);
 \end{scope} 
  
 \begin{scope}[xshift=6.5cm]  
  \node at (0.5,4) {Case 3a:};
  \node[vertex] (w0) at (0.5,2.2) {};
  \node[vertex] (w1) at (0.5,3.2) {};
  \draw[j] (w0) to (w1); 
  \node[vertex] (v00) at (0,0.6) {};
  \node[vertex] (v10) at (1,0.6) {};
  \draw[j] (v00) to (v10);
  \node[vertex] (v01) at (0,1.6) {};
  \node[vertex] (v11) at (1,1.6) {};
  \draw[j] (v01) to (v11);
  \draw[edge] (v00) to (v01);
  \draw[edge] (v10) to (v11);
  \draw[edge] (v00) to (v11);
  \draw[cut] (0,0.6) circle (.25 and .25);
  \draw[take] (0,1.6) circle (.25 and .25);
  \draw[take] (1,0.6) circle (.25 and .25);
  \draw[cut] (1,1.6) circle (.25 and .25);
  \draw[cut] (0.5,1.2) circle (1.2 and 1.5);
  \node[red] at (-0.4,2.6) {$\bar U$};
  \draw[edge] (v00) to (-0.9,0.5);
  \draw[edge] (v01) to (-0.9,1.7);
  \draw[edge] (v10) to (1.9,0.5);
  \draw[edge] (v11) to (1.9,1.7);
  \draw[edge] (v01) to (w0);
  \draw[edge] (v11) to (w0);
 \end{scope} 
  
 \begin{scope}[xshift=10.5cm]  
  \node at (0.5,4) {Case 3b:};
  \node[vertex] (w0) at (0.5,2.2) {};
  \node[vertex] (w1) at (0.5,3.2) {};
  \draw[j] (w0) to (w1); 
  \node[vertex] (v00) at (0,0.6) {};
  \node[vertex] (v10) at (1,0.6) {};
  \draw[j] (v00) to (v10);
  \node[vertex] (v01) at (0,1.6) {};
  \node[vertex] (v11) at (1,1.6) {};
  \draw[j] (v01) to (v11);
  \draw[edge] (v00) to (v01);
  \draw[edge] (v10) to (v11);
  \draw[edge] (v00) to (v11);
  \draw[cut] (0,0.6) circle (.25 and .25);
  \draw[cut] (0,1.6) circle (.25 and .25);
  \draw[cut] (1,0.6) circle (.25 and .25);
  \draw[take] (1,1.6) circle (.25 and .25);
  \draw[take] (0.5,1.2) circle (1.2 and 1.5);
  \node[red] at (-0.4,2.6) {$\bar U$};
  \draw[edge] (v00) to (-0.9,0.5);
  \draw[edge] (v01) to (-0.9,1.7);
  \draw[edge] (v10) to (1.9,0.5);
  \draw[edge] (v10) to[out=45,in=315] (1.3,1.9);
  \draw[edge] (1.3,1.9) to[out=135,in=45] (v01);
  \draw[edge] (v01) to[bend left] (w0);
 \end{scope} 
    
    \end{tikzpicture}
  \end{center}    
  \caption{The different cases in the proof of Lemma \ref{lemma:cutpackingfromhalfintegral}.
  \label{fig:packingTcuts}
  }
\end{figure}

\smallskip\noindent
{\bf Case 1:} The elements of $\Lscr$ are pairwise disjoint. 
Let $\Lscr'\subseteq \Lscr$ be the solution returned. 
$\Lscr'$ corresponds to an independent set in $P$, so it is stable 
and has size at least $|V(P)|/4=|\Lscr|/4$. 
Notice that when $\Lscr=\emptyset$, case 1 applies, and the empty set is returned.

\smallskip\noindent
{\bf Case 2:} Some $U_1,U_2\in\Lscr$ are equal. 
 We assume by the induction hypothesis that {\sc CutPacking}$(G,\Lscr\setminus\{U_1,U_2\})$ is stable and 
 {\sc CutPacking}$(G,\Lscr\setminus\{U_1,U_2\})$ has size at least $(|\Lscr|-2)/4$. 
 First, the solution returned has size $\frac{|\Lscr|-2}{4}+1=\frac{|\Lscr|+2}{4}>\frac{|\Lscr|}{4}$.
 To prove that $\{U_1\}\cup \textsc{CutPacking}(G,\Lscr\setminus\{U_1,U_2\})$ is stable, 
 we only have to check that $\delta(U_1)\cap \delta(U)=\emptyset$ for all $U\in \Lscr\setminus\{U_1,U_2\}$. 
 This holds since any edge in $\delta(U_1)$ is also contained in $\delta(U_2)$ and is contained in at most two of the cuts.
 
\smallskip\noindent
{\bf Case 3:} There exists a set $\overline{U}\in \Lscr$ that is not minimal but all sets $U_1,\dots, U_l\in\Lscr$ 
that are proper subsets of $\overline{U}$ are minimal. 
Therefore the sets $U_i$ ($i=1,\ldots,l$) are pairwise disjoint.     
Then consider the graph $P$ constructed in step 7 of the algorithm and a partition of its vertex set into four independent sets.
Now we consider two subcases.

If one of these independent sets contains at least $\frac{l+1}{4}$ normal vertices (case 3a), 
then by induction hypothesis, the solution has size at least $\frac{l+1}{4}+\frac{|\Lscr|-(l+1)}{4}=\frac{|\Lscr|}{4}$.
Otherwise (case 3b) all these independent sets contain exactly $\frac{l}{4}$ normal vertices and we break ties by taking the one that contains the special vertex. 
Using the induction hypothesis again, the solution returned has size at least $\frac{l}{4}+\frac{|\Lscr|-l}{4}=\frac{|\Lscr|}{4}$. 

After using induction hypothesis, this solution is stable because any edge in $\delta(U_i)\cap\delta(U)$ for $U\in \Lscr$ 
such that $U\subseteq V\setminus\overline{U}$ would also be contained in $\delta(\overline{U})$, but no edge is contained in three of the cuts. 
If $\overline{U}\in\textsc{CutPacking}(\Lscr\setminus \{U_1,\dots,U_l\})$, 
then the special vertex is in the independent set $\mathcal{K}$, and hence $\delta(\overline{U})\cap \delta(U_i)=\emptyset$ for all $U_i\in\Lscr'$. 
 \end{proof}

\section{Max-Multiflow Min-Multicut Ratio (Proof of Theorem \ref{thm:mainResultTwo})}
\label{section:LPs}

\subsection{Linear Programming Relaxations and Integrality Gaps}

$\EDP$ consists of finding as many $D$-cycles as possible.
If $\Cscr$ denotes the set of all $D$-cycles, then we
look for a maximum subset of $\Cscr$ whose elements are pairwise disjoint.
The natural fractional relaxation is (along with its linear programming dual):

\begin{figure}[H]
\begin{minipage}[b]{0.49\linewidth}\centering
\begin{align*}
 \max \sum_{C \in \mathcal{C}} f_C &  \\
\sum_{C \in \mathcal{C}: e\in C} f_C  \leq  1 & \qquad \forall e \in D\cupp E\\
f_C  \geq   0  & \qquad \forall C \in \mathcal{C}. 
\end{align*}
\end{minipage}
\begin{minipage}[b]{0.49\linewidth}\centering
\vspace*{-0.25in}
\begin{align*}
\min \sum_{e \in D\cupp E} y_e  & \\
\sum_{e \in C} y_e \geq  1  &  \qquad \forall C \in \mathcal{C}\\
y_e  \geq   0  & \qquad \forall e \in D\cupp E 
\end{align*}
\end{minipage}
\caption{The $D$-cycle packing LP and its dual, the $D$-cycle covering LP.}
\label{fig:LP_Cycle}
\end{figure}

If we go to the planar dual graph and let $\Uscr$ denote the set of vertex sets inducing $D$-cuts, these LPs become
\begin{figure}[H]
\begin{minipage}[b]{0.5\linewidth}\centering
\begin{align*}
 \max \sum_{U \in \mathcal{U}} f_U &  \\
\sum_{U \in \mathcal{U}: e\in\delta(U)} f_U  \leq  1&   \qquad \forall e \in D\cupp E\\
f_U  \geq   0  & \qquad \forall U \in \mathcal{U}. 
\end{align*}
\end{minipage}
\begin{minipage}[b]{0.49\linewidth}\centering
\vspace*{-0.25in}
\begin{align*}
\min \sum_{e \in D\cupp E} y_e  & \\
\sum_{e \in \delta(U)} y_e \geq  1  &  \qquad \forall U \in \mathcal{U}\\
y_e  \geq   0  & \qquad \forall e \in D\cupp E 
\end{align*}
\end{minipage}
\caption{The $D$-cut packing LP and its dual, the $D$-cut covering LP.}
\label{fig:LP_EDP}
\end{figure}

It turns out that the $D$-cut packing LP is equivalent to the nonnegative cycle LP. We need the following well-known 
generalization of Theorem \ref{thm:2packingofTcuts} (see, e.g., \cite{Sebo1997} and the references therein):

\begin{theorem}
\label{thm:algorithmicpackingofTcuts}
For every graph $G=(V,E)$ with capacities $u:E\to\mathbb{R}_{\ge 0}$
and every set $J$ of edges such that $u(C\cap J)\le u(C\setminus J)$
for every cycle $C$, 
one can compute in strongly polynomial time a laminar family $\Lscr$ of vertex sets
with weights $w:\Lscr\to\mathbb{R}_{>0}$ such that
$|\delta(U)\cap J|=1$ for all $U\in\Lscr$,
$\sum_{U\in\Lscr}w(U)=\sum_{e\in J} u(e)$, and
$\sum_{U\in\Lscr: e\in\delta(U)} w(U) \le u(e)$ for all $e\in E$.
If $u$ is integral, than $w$ can be chosen half-integral.
\end{theorem}

Using this, we can show:

\begin{lemma}
\label{lemma:LPsequivalent}
The $D$-cut packing LP is equivalent to the nonnegative cycle LP~(\ref{LP:NNC}). 
Their values are the same, from an optimum solution to the former we can get an optimum solution to the latter in polynomial time, and vice versa.
\end{lemma}

\begin{proof}
For any feasible solution $f$ of the $D$-cut packing LP define a vector $x$ by
$x_e:=$ $\sum_{U\in\Uscr:e\in\delta(U)} f_U$ for $e\in D$.
Then $x$ is a feasible solution to the nonnegative cycle LP because for every cycle $C$ in $G+H$
we have 
$$\sum_{e\in C\cap D} x_e 
= \hspace{-6pt}\sum_{e\in C\cap D} \sum_{\substack{U\in\Uscr:\\e\in\delta(U)}} f_U
= \sum_{U\in\Uscr} f_U |C\cap D\cap\delta(U)|
\le \sum_{U\in\Uscr} f_U |C\cap E\cap\delta(U)|
= \hspace{-6pt} \sum_{e\in C\cap E} \sum_{\substack{U\in\Uscr:\\e\in\delta(U)}} f_U
\le |C\cap E|.$$
Here the first inequality holds because a cycle $C$ and a cut $\delta(U)$ 
intersect in an even number of edges, 
and for $U\in\Uscr$ at most one edge in the intersection belongs to $D$.
We also have $\sum_{e\in D}x_e=\sum_{U\in\Uscr}f_U$.

Conversely, let $x$ be a feasible solution to the nonnegative cycle LP.
Define $u(e):=x_e$ for $e\in D$ and $u(e):=1$ for $e\in E$.
By Theorem \ref{thm:algorithmicpackingofTcuts}
(applied to $G+H$ and $J:=D$), one can compute
a laminar family $\Lscr$ of vertex sets
with weights $w:\Lscr\to\mathbb{R}_{>0}$ such that
$|\delta(U)\cap D|=1$ for all $U\in\Lscr$,
$\sum_{U\in\Lscr}w(U)=\sum_{e\in D} u(e)$, and
$\sum_{U\in\Lscr: e\in\delta(U)} w(U) \le u(e)$ for all $e\in D\cupp E$.
Set $f_{U}:=w(U)$ for $U\in\Lscr$
and $f_U:=0$ otherwise. 
Then $f$ is a feasible solution to the $D$-cut packing LP
with $\sum_{U\in\Uscr}f_U=\sum_{U\in\Lscr}w(U)=\sum_{e\in D} u(e)=\sum_{e\in D}x_e$.
\end{proof}

This implies:

\begin{corollary}
\label{cor:primalgap}
If $G+H$ is planar, 
the integrality gap of the $D$-cut packing LP (and hence the integrality gap of the $D$-cycle packing LP) is at least 2 and at most 32.
\end{corollary}

\begin{proof}
We showed in the proof of Theorem \ref{thm:mainResultOne} that there is an 
algorithm that computes an integral solution of the $D$-cut packing LP of at least $\frac{1}{32}$ times the value of the nonnegative cycle LP.
By Lemma \ref{lemma:LPsequivalent} this implies the upper bound.

The lower bound of 2 is attained for $G+H=K_4$ (the complete graph on 4 vertices) if $D$ is a perfect matching in $G+H$, see Figure \ref{fig:lb} (a).
\end{proof}

For the nonnegative cycle LP the integrality gap could be smaller. We know that it is between $\frac{3}{2}$ (shown by $G+H=K_4$ and $D=\delta(u)$ for an arbitrary vertex $u$, see Figure \ref{fig:lb} (b)) and 
16 (shown by Theorem \ref{thm:s1ands2}).

We now observe that the dual LPs have integrality gap at most 2.

\begin{lemma}
\label{lemma:dualgap}
The integrality gap of the $D$-cut covering LP (and hence the integrality gap of the $D$-cycle covering LP if $G+H$ is planar) is at most 2, 
and it is at least $\frac{3}{2}$ even if $G+H$ is planar.
\end{lemma}

\begin{proof}
The lower bound follows from an example by Cheriyan, Karloff, Khandekar, and Könemann \cite{Cheriyan2008} for the tree augmentation problem
(let $D$ contain the tree edges and $E$ contain the links of this example, then their LP is equivalent to the $D$-cut covering LP). See Figure \ref{fig:lb} (c).

We now show the upper bound.
The $D$-cut covering LP is equivalent to
\begin{align}
\label{eq:sndplp}
\min \sum_{e \in D\cupp E} y_e  & \\
\sum_{e \in \delta(U)} y_e + \sum_{e\in \delta(U)\cap D} z_e \geq  2  &  \qquad \forall U\subseteq V \text{ with } D\cap\delta(U)\not=\emptyset \notag \\
y_e  \geq   0  & \qquad \forall e \in D\cupp E \notag \\
0\le z_e \le 1 &  \qquad \forall e \in D \notag
\end{align}
because in \eqref{eq:sndplp} we can assume that $z_e=1$ for all $e\in D$, and then it is exactly the same as the $D$-cut covering LP.
Now \eqref{eq:sndplp} is a survivable network design LP of the type
\begin{align*}
\min \sum_{e\in E'} c(e) x_e & \\
\sum_{e \in \delta(U)} x_e \geq  f(U)  &  \qquad \forall U\subseteq V \notag \\
0 \le x_e \le 1  & \qquad \forall e \in E' \notag
\end{align*}
for the graph whose edge set $E'$ consists of an edge $e$ of cost $c(e)=1$ for each $e\in D\cupp E$ 
and another parallel edge $e'=\{v,w\}$ with $c(e')=0$ for each $\{v,w\}\in D$. 
The requirement function $f:2^V\to\mathbb{Z}_{\ge 0}$ is
given by $f(U)=2$ if $D\cap\delta(U)\not=\emptyset$ and $f(U)=0$ otherwise.
This function $f$ is easily seen to be proper (and thus weakly supermodular), and
hence Jain's iterative rounding theorem \cite{Jain2001} tells that there is an
integral feasible solution $(y,z)$ to \eqref{eq:sndplp} of cost at most twice the LP value.
Then $y$ is an integral feasible solution to the $D$-cut covering LP of the same cost.
\end{proof}

\begin{figure*}[t!]
    \centering
    \begin{subfigure}[t]{0.175\textwidth}
        \centering
        \includegraphics[height=1in]{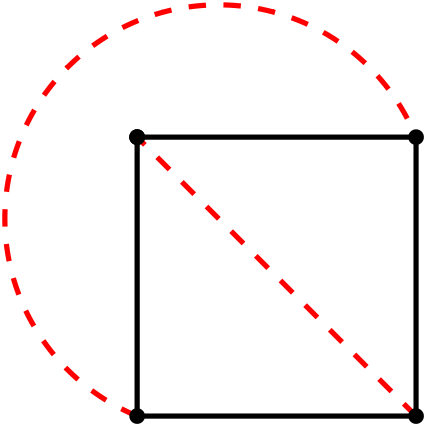}
        \caption{The integrality gap of the $D$-cut packing LP is at least $2$. 
        }
    \end{subfigure}
    \hfill
    \begin{subfigure}[t]{0.175\textwidth}
        \centering
        \includegraphics[height=0.9in]{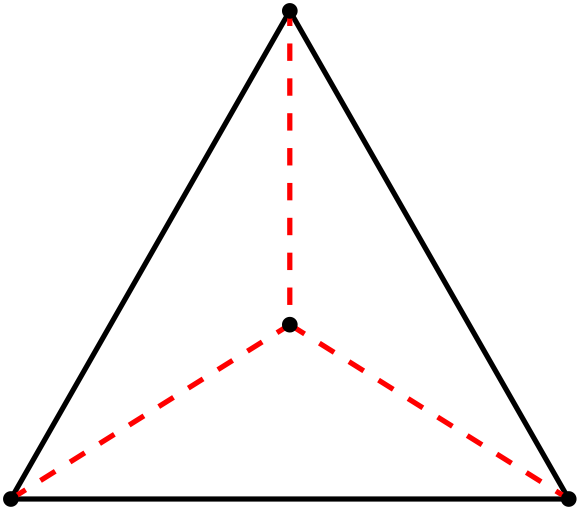}
        \caption{The integrality gap of the nonnegative cycle LP is at least $\frac{3}{2}$. 
        }
    \end{subfigure}
    \hfill
    \begin{subfigure}[t]{0.58\textwidth}
        \centering
        \includegraphics[width=\textwidth]{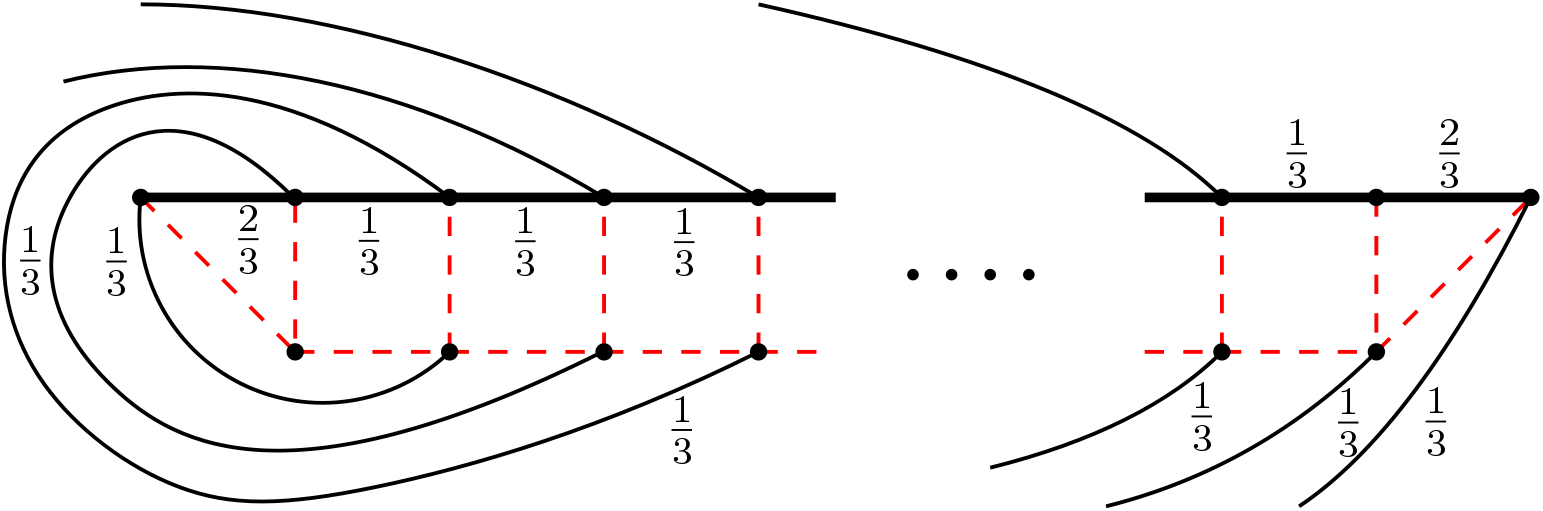}
        \caption{The integrality gap of the $D$-cut covering LP is at least $\frac{3}{2}$.   
        Bold black edges form an optimum integral solution. 
        Values $\frac{1}{3}$ and $\frac{2}{3}$ on the black edges form a feasible fractional solution.}
    \end{subfigure}
    \caption{Examples for lower bounds of integrality gaps. Solid black edges belong to $E$, dashed red edges belong to $D$. }
    \label{fig:lb}
\end{figure*}

We remark that even if just $G$ is planar (not $G+H$) the integrality
gap of the $D$-cycle covering LP is bounded by a (large) constant \cite{tardos1993}. 

\subsection{Multiflows and Multicuts}

A natural generalization of $\EDPs$ takes as input,
in addition to $G$ and $H$, a capacity function
$u:D\cupp E\to\mathbb{Z}_{>0}$ and asks for a maximum
number of $D$-cycles such that each edge $e\in D\cupp E$ belongs to at most $u(e)$ of them.
If $u(e)=1$ for all $e\in D\cup E$, this is $\EDPs$.

This generalization reduces to $\EDPs$ by replacing each edge $e$ by $u(e)$ parallel edges of unit capacity. 
In the planar dual, this corresponds to replacing $e$ by a path with $u(e)$ edges. 
If $u$ is given in binary representation, this reduction should of course not be performed explicitly.
Nevertheless our algorithm can easily be generalized to run in polynomial time for general capacities. 
A dual of a weighted planar graph can be coded by a weighted planar graph where an edge $e$ of weight $w_e$ represents a path of length $w_e$. 
All the algorithms involved in proving Theorem \ref{thm:s1ands2} can be easily extended to weighted instances. 
Once we have a weighted feasible solution $J$ that respects the non-negative cycle condition, we apply Theorem \ref{thm:algorithmicpackingofTcuts} 
to get a (weighted) laminar family that satisfies the requirements of Lemma \ref{lemma:cutpackingfromhalfintegral}. 
Then, the final feasible solution of EDP can be represented by an integral flow of a certain value for some demand edges. 
These flows can be decomposed into (and hence represented by) at most $|E|$ paths by standard flow decomposition.

If $u(e)=\infty$ (or large enough) for all $e\in D$, the problem has been called the maximum integer multiflow problem. 
(Note that $u(e)=\infty$ for all $e\in D$ can in fact be assumed without loss of generality because a demand edge $e=\{v,w\}$ with capacity $u(e)$ 
can be replaced equivalently by a demand edge $\{v,x\}$ of infinite capacity and a supply edge $\{x,w\}$ of capacity $u(e)$, where $x$ is a new vertex.)

If we do not require integrality, the maximum multiflow problem is of course a linear program:

\begin{figure}[H]
\begin{minipage}[b]{0.5\linewidth}\centering
\begin{align*}
 \max \sum_{C \in \mathcal{C}} f_C &  \\
\sum_{C \in \mathcal{C}, C \ni e} f_C  \leq  u(e) & \qquad \forall e \in E\\
f_C  \geq   0  & \qquad \forall C \in \mathcal{C}. 
\end{align*}
\end{minipage}
\hspace{0.5cm}
\begin{minipage}[b]{0.5\linewidth}\centering
\vspace*{-0.25in}
\begin{align*}
\min \sum_{e \in E} u(e) y_e  & \\
\sum_{e \in C\cap E} y_e \geq  1  &  \qquad \forall C \in \mathcal{C}\\
y_e  \geq   0  & \qquad \forall e \in E 
\end{align*}
\end{minipage}
\caption{The multiflow LP and its dual, the multicut LP.}
\label{fig:LP_Multiflow}
\end{figure}

The feasible solutions to the multiflow LP are called multiflows (here we use the form after flow decomposition).
The integral feasible solutions to the dual LP correspond to edge sets $F\subseteq E$ such that deleting $F$ destroys all $D$-cycles. They are called \emph{multicuts}. The \emph{capacity} of a multicut is the total capacity of its edges.
Of particular interest is the worst ratio of the minimum capacity of a multicut and the maximum value of an integer multiflow.
The (integer version of the) famous max-flow min-cut theorem says that this ratio is 1 if $|D|=1$. 
However in general, even if $G$ is sub-cubic and planar the ratio can be as large as $\Theta(|D|)$ \cite{Garg1997}.
Besides when $G$ is a tree (then the ratio is 2 \cite{Garg1997}), when $G$ is planar and has bounded tree-width \cite{Bentz2006} or when $G+H$ is series-parallel \cite{Cornaz11} (then the ratio is $1$), very few cases are known where the ratio can be bounded by a constant.
We can now show a constant upper bound when $G+H$ is planar. This is Theorem \ref{thm:mainResultTwo}, which we restate here:

\mainresulttwo*

\begin{proof}
The integrality gap of the multiflow LP equals the integrality gap of the $D$-cycle packing LP, as the two problems can be reduced to each other:  the gap cannot be smaller because given an instance of the $D$-cycle packing LP, as said above, we can replace a demand edge $\{v,w\}$ by a demand edge $\{v,x\}$ of infinite (or very large) capacity and a supply edge $\{x,w\}$ of capacity 1 to form a multiflow instance. 

Conversely, the gap also cannot be larger since we can reduce the multiflow problem to $D$-cycle packing by replacing every edge $e$ by $u(e)$ parallel edges with unit capacity. 
Note that both reductions preserve planarity.
The integrality gap of the multicut LP equals the integrality gap of the $D$-cycle covering LP by the same argument.
Now Corollary \ref{cor:primalgap} and Lemma \ref{lemma:dualgap} imply the assertion.
\end{proof}

Determining the exact integrality gaps remains an open question.
Without the planarity assumption, all the LPs except for the $D$-cut covering LP have unbounded integrality gap
(for the $D$-cycle packing LP, a well-known example in \cite{Garg1997} shows that 
the gap is in the order of $\Omega(\sqrt{n})$, even when $G$ is planar, subcubic, and all demand pairs lie in the boundary of the outer face of $G$;
for the $D$-cut packing LP and the nonnegative cycle LP, consider $G+H=K_n$ and $D=\delta(v)$ for some vertex $v$;
for the $D$-cycle covering LP, let $G$ be a bounded-degree expander graph with $n$ vertices and $D$ the set of $\frac{n^2}{4}$ vertex pairs with largest distance \cite{Garg1996}).

\section{NP-completeness of \NNC}\label{section:NNCisNPcomplete}

In this section we prove that \NNC is NP-hard. 
In fact, we consider the \NNC decision problem, which takes as input an instance of \NNC and an integer $k$, 
and asks whether there exists a solution of cardinality $k$.

\begin{theorem}
The \NNC decision problem is NP-complete even when $G+H$ is planar.
\label{theorem:NNCisNPcomplete}
\end{theorem}

To prove membership in NP, here is a polynomial-time algorithm that, given an instance $G=(V,E)$ and $H=(V,D)$ 
and a subset $D'\subseteq D$ of cardinality $k$, verifies whether $D'$ is a solution: 
As in the proof of Lemma~\ref{lemma:solvelp}, assign weight $1$ to edges of $E$, $-1$ to edges of $D'$, and $0$ to edges of $D\setminus D'$; 
compute a minimum-weight $\emptyset$-join; then $D'$ is a solution if and only if that minimum weight is 0. 

To prove NP-completeness, recall the vertex cover problem:
given a graph $G$ and an integer $k$, it asks whether $G$ has a vertex cover of size $k$.
This problem is well-known to be NP-complete even in planar graphs~\cite{GareyJohnsonStockmeyer}. 
We give a polynomial-time transformation from that problem.

\begin{figure}
    \centering
    \includegraphics[width=14.5cm]{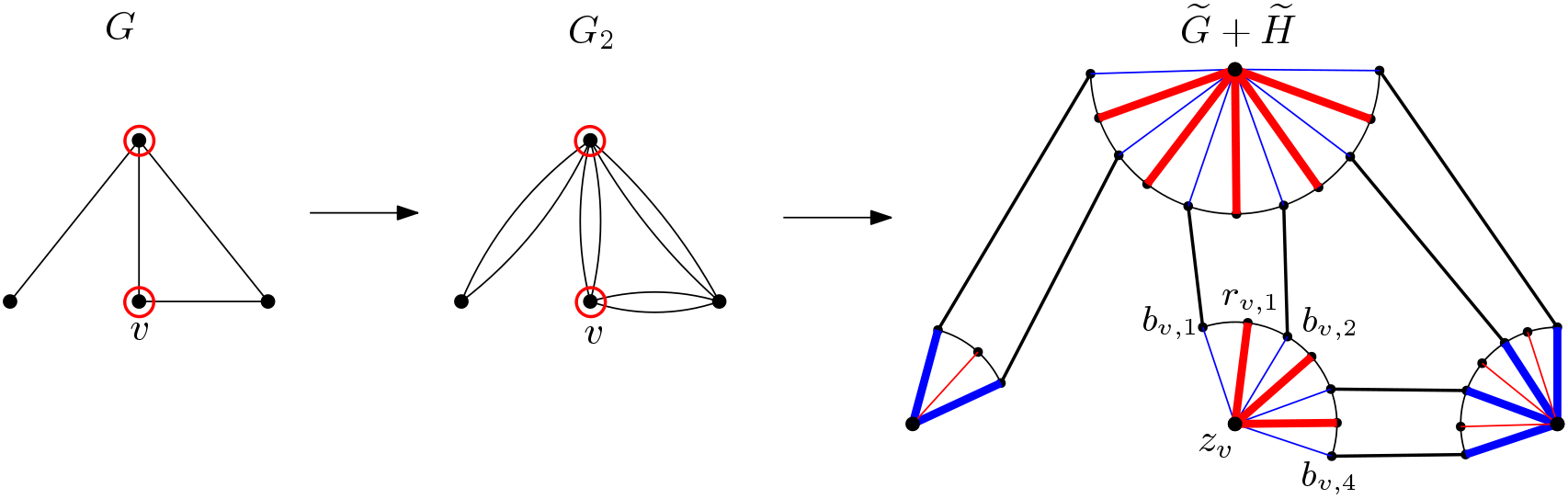}
    \caption{The reduction from vertex cover in planar graphs. 
    On the left-hand side picture, the vertices surrounded by a circuit form a minimum vertex cover $X$ in $G$. 
    On the right-hand side, the supply edges are shown in black, and the demand edges are colored (red or blue). 
    The solution to $\NNCs$  corresponding to $X$ is the set of thicker edges. 
    It is obtained by taking the red edges in each gadget associated to a vertex in $X$ and the blue edges otherwise.}
    \label{fig:nncnphard}
\end{figure}

Let $G=(V,E)$ be a planar graph. 
Let $G_2:=(V,E \cupp  E)$ denote the (multi)graph  where each edge of $G$ is duplicated. 
We construct an instance  $\widetilde{G}+\widetilde{H}=(\widetilde{V}, \widetilde{E}\cupp \widetilde{D})$ 
of \NNC as follows, illustrated in Figure \ref{fig:nncnphard}.

For each vertex $v\in V$ of degree $d=d_G(v)$ in $G$, there is a  \emph{gadget} $(\widetilde{V}_v, \widetilde{E}_v\cupp \widetilde{D}_v)$  built as follows. 
The vertex set $\widetilde{V}_v$ consists of  $4d$ vertices $z_v,b_{v,1},r_{v,1},b_{v,2},\dots, r_{v,2d-1},b_{v,2d}$. 
The supply edge set $\widetilde{E}_v$ consists of edges $\{b_{v,i},r_{v,i}\}$ and $\{r_{v,i},b_{v,i+1}\}$, for $1\le i \le 2d-1$.
The demand edge set $\widetilde{D}_v$ consists of $\widetilde{B}_v\cup \widetilde{R}_v$, 
where  $\widetilde{B}_v:=\{\{z_v,b_{v,i}\} \mid 1\le i \le 2d\}$ and $\widetilde{R}_v:=\{\{z_v,r_{v,i}\} \mid 1\le i \le 2d-1\}$.

To complete the construction of $(\widetilde{G},\widetilde{H})$, starting from a planar embedding of $G_2$, 
replace each vertex $v$ by the corresponding gadget, and each edge $e=\{u,v\}$ in $G_2$ by a supply edge $\{b_{u,i},b_{v,j}\}$ in $\widetilde{E}$, 
where $i$ and $j$ are chosen so that $\widetilde{G}+\widetilde{H}$ is planar. 

The construction of $\widetilde{G}+\widetilde{H}$ can be done in polynomial time. 
Hence the following lemma completes the proof of Theorem \ref{theorem:NNCisNPcomplete}.

\begin{lemma}
Let $G$ be a graph.
There exists a vertex cover in $G$ of size $k$ if and only if there exists a subset $D'\subseteq \widetilde{D}$ of size $4|E|-k$ 
that is a feasible solution to \NNC in $\widetilde{G}+\widetilde{H}$.
\end{lemma}

\begin{proof}
($\Rightarrow$) Let $X$ be a vertex cover of size $k$ in $G$. Define 
$D':= \bigl(\bigcup_{v\in X} \widetilde{R}_v \bigr)\cup \bigl(\bigcup_{v\notin X} \widetilde{B}_v\bigr)$. 
It is easy to check that $D'$ has size $4|E|-k$. We show now that $D'$ is a feasible solution for \NNC. 
Let $\widetilde{I}$ denote the set of $|D'|$ vertices naturally associated to $D'$, taking the endpoint of each edge of $D'$ that is not $z_v$ for any $v\in V$. 
By construction and since $X$ is a vertex cover in $G$, $\widetilde{I}$ is an independent set in $\widetilde{G}+\widetilde{H}$. 
 Hence $\bigl\{\delta(\{\widetilde{v}\}) \mid \widetilde{v}\in\widetilde{I}\bigr\}$ 
 is a $D'$-cut packing.
 For each cycle $\widetilde{C}$ in $\widetilde{G}+\widetilde{H}$, and for each $\widetilde{v} \in \widetilde{I}$, we have
 $|\widetilde{C}\cap \delta(\{\widetilde{v}\})\cap D'|\le |\widetilde{C}\cap \delta(\{\widetilde{v}\}) \cap \widetilde{E}|$ and thus
 $$
 |\widetilde{C}\cap D'|=\sum_{\widetilde{v}\in\widetilde{I}}|\widetilde{C} \cap \delta(\{\widetilde{v}\})\cap D'|\le
 \sum_{\widetilde{v}\in\widetilde{I}}|\widetilde{C}\cap \delta(\{\widetilde{v}\}) \cap\widetilde{E}|\le |\widetilde{C}\cap\widetilde{E}|.
 $$

($\Leftarrow$) Assume now that we are given a feasible solution $D'$ for \NNC in $\widetilde{G}+\widetilde{H}$ of size $4|E|-k$. 
We define $X:=\{v\in V \mid D'\cap \widetilde{D}_v\not=\widetilde{B}_v\}$. 
We show that $X$ is a vertex cover in $G$ with size at least $k$. 
If there was an edge $\{u,v\}\in E$ with $u\notin X$ and $v\notin X$, then there would be a cycle, 
induced by vertices  $\{z_u, b_{u,i}, b_{v,j},z_v,b_{v,j+1},b_{u,i+1}\}$ in $\widetilde{G}+\widetilde{H}$ for some indices $i,j$, 
that contains four edges in $D'$ but only two supply edges. We now prove that $|X|\le k$.

Since $D'$ is a feasible solution, we have $|D'\cap \widetilde{D}_v|\le 2d_G(v)$ for all $v\in V$, 
for otherwise one could find a triangle in $\widetilde{G}_v+\widetilde{H}_v$ with two edges in $D'$.
Moreover, $|D'\cap \widetilde{D}_v| = 2d_G(v)$ only if $D'\cap \widetilde{D}_v=\widetilde{B}_v$, i.e., only if $v\notin X$. 
This implies
$$|X|\le \sum_{v\in V} \left( 2d_G(v)-|D'\cap\widetilde{D}_v| \right) \le 4|E| - |D'| = k.
$$
\end{proof}

\section{Conclusion}

We designed the first constant-factor approximation algorithm for \EDP if $G+H$ is planar.

If all demand edges lie on a single face of a planar embedding of $G$, then the planar dual $H^*$ has only one nontrivial connected component (and isolated vertices).
In this case, any instance satisfying the cut condition has a complete solution \cite{KorachPenn1992}. 
Then a 3-coloring of the outerplanar graph $G[L]$ in Algorithm~\ref{algo:S1} and a small modification of Algorithm~\ref{algo:S2} 
(reducing the upper bounds on $y$ from 1 and 2 to $\frac{1}{2}$ and 1) yields a 4-approximation.
If we could solve \NNC in this case, we would even obtain an exact algorithm. 
This motivates the following open question: can \NNC  be solved in polynomial time if $H$ has only one nontrivial connected component?
%

One attempt, also to reduce the constant factor that we lose when rounding a solution to the 
nonnegative cycle LP~(\ref{LP:NNC}) in general, would be to find a characterisation of optimal LP solutions. 
However, this seems to be difficult. Figure \ref{fig:nothalfintegral} shows that
there are instances in which the unique optimum LP solution is not half-integral.

\begin{figure}
    \centering
    \includegraphics[width=0.6\textwidth]{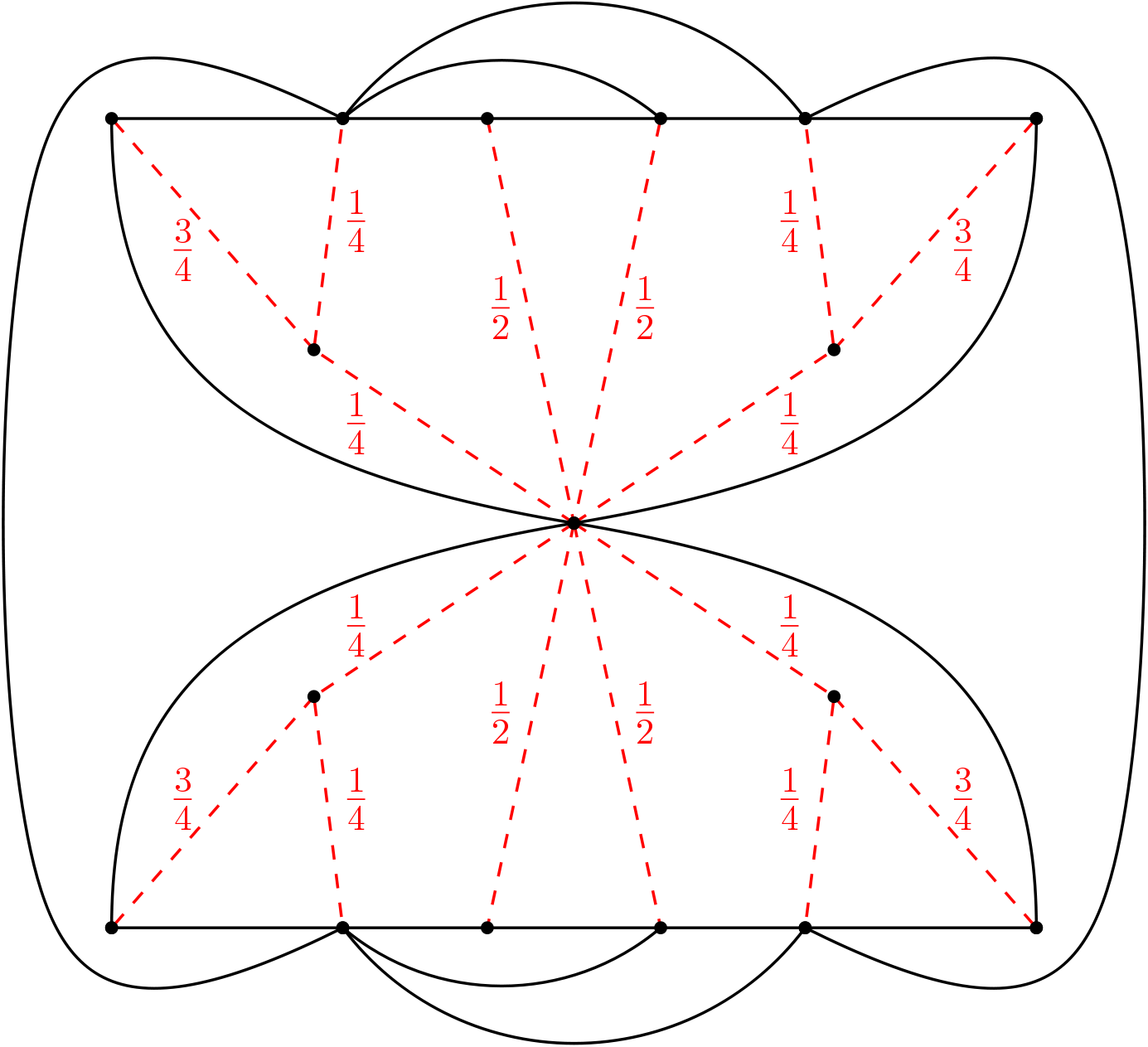}
    \caption{This instance shows that the nonnegative cycle LP~(\ref{LP:NNC}) does not always have half-integral optimal solutions. 
    Black (solid) edges are supply edges, red (dashed) edges represent demand edges together with their value in the unique optimum fractional solution.}
    \label{fig:nothalfintegral}
\end{figure}

\bibliographystyle{siam}
\bibliography{bibliography}


\end{document}